%% file: main-arxiv.tex
\documentclass[sigconf, nonacm]{acmart}

\AtBeginDocument{%
  }

\setcopyright{none}
\copyrightyear{}
\acmYear{}
\acmDOI{}

\acmPrice{}
\acmISBN{}

\usepackage{amsfonts}
\usepackage{amsthm}
\usepackage{tabularx}
\usepackage{hyperref}
\usepackage{graphicx, subcaption}
\usepackage{bbold}
\usepackage{algorithm}
\usepackage[noend]{algorithmic}

\usepackage{balance}	

\usepackage{tikz}
\usetikzlibrary{calc}
\usepackage{enumitem}
\usepackage{notation}

\begin{document}

\title{Overlapping and Robust Edge-Colored Clustering in Hypergraphs}



\author{Alex Crane}
\email{alex.crane@utah.edu}
\affiliation{%
	\institution{University of Utah}
	\city{Salt Lake City}
	\state{Utah}
	\country{USA}
}

\author{Brian Lavallee}
\email{brian.lavallee@utah.edu}
\affiliation{%
	\institution{University of Utah}
	\city{Salt Lake City}
	\state{Utah}
	\country{USA}
}

\author{Blair D. Sullivan}
\email{sullivan@cs.utah.edu}
\affiliation{%
	\institution{University of Utah}
	\city{Salt Lake City}
	\state{Utah}
	\country{USA}
}

\author{Nate Veldt}
\email{nveldt@tamu.edu}
\affiliation{%
	\institution{Texas A\&M University}
	\city{College Station}
	\state{Texas}
	\country{USA}
}

\renewcommand{\shortauthors}{Alex Crane, Brian Lavallee, Blair D. Sullivan, and Nate Veldt}

\input{sections/abstract}

\maketitle

\input{sections/intro-new}
\input{sections/prelims}
\input{sections/approximations}
\input{sections/hardness}
\input{sections/experiments}

\input{sections/conclusion}
\begin{acks}
This work was supported in part by the Gordon \& Betty Moore Foundation under award GBMF4560 to Blair D. Sullivan.
\end{acks}


\bibliographystyle{ACM-Reference-Format}
\bibliography{refs}


\end{document}

%% file: sections/abstract.tex
\begin{abstract}
    A recent trend in data mining has explored (hyper)graph clustering algorithms for
    data with categorical relationship types. Such algorithms have 
    applications in the analysis of social, co-authorship, and protein
    interaction networks, to name a few. Many such applications
    naturally have some overlap between clusters, a nuance which is missing
    from current combinatorial models. Additionally, existing models lack a
    mechanism for handling noise in datasets. We address both of these
    concerns by generalizing \ECCfull{}, a recent framework for 
    categorical clustering of hypergraphs. Our generalizations allow for a budgeted 
    number of either (a) overlapping cluster assignments or (b) node deletions. 
    For each new model we present a greedy algorithm which approximately minimizes
    an edge mistake objective, as well as 
    bicriteria approximations where the second approximation factor 
    is on the budget. Additionally, we address the parameterized complexity 
    of each problem, providing FPT algorithms and hardness results.
\end{abstract}

%% file: sections/intro-new.tex
\section{Introduction}
Graph clustering is a fundamental problem in data mining, relevant whenever the dataset captures relationships between entities. A recent trend has focused on clustering hypergraphs~\cite{papa2007hypergraph,li2017motif,li2017inhomogeneous,gleich2018correlation,li2018submodular,fukunaga2019lp,fountoulakis2021local,hein2013total}, which better capture \emph{multiway} relationships such as group social interactions, group email correspondence, and academic co-authorship.
A second trend has emphasized clustering on edge-colored (hyper)graphs, where the
color of a (hyper)edge indicates an interaction of a certain
\emph{type} or \emph{category}
\cite{amburg2020clustering,bonchi2015chromatic,anava2015improved,klodt2021color,xiu2022chromatic,amburg2022diverse, veldt2022optimal}.

Given a (hyper)graph with edge categories (modeled by colors), \emph{categorical clustering} seeks to cluster nodes in such a way that node clusters tend to match with edge categories. This type of clustering is natural in many different settings. For example, a hyperedge can be used to represent a group of researchers who co-author a paper together, and edge colors can denote paper discipline or venue (e.g., WSDM or SIGCOMM)~\cite{amburg2020clustering,bonchi2012chromatic}. In this context, categorical clustering provides a way to infer authors' research areas based on publication history. Edge-colored hypergraphs can also be used to represent sets of ingredients (nodes) in the same food recipe (hyperedge), with edge color indicating cuisine type (e.g., Korean food or French cuisine)~\cite{amburg2020clustering,klodt2021color,xiu2022chromatic}. Categorical clustering then provides a way to identify groups of ingredients that are frequently used together in the same cuisine type. Edge colors can also represent discrete time windows in temporal networks~\cite{amburg2020clustering}, in which case categorical clustering can be used for temporal clustering, e.g., to identify groups of users in an online social network that are active in the same time period~\cite{amburg2020clustering}. Existing tools for categorical clustering have also been used to cluster biological networks (where colors indicate gene interaction types)~\cite{bonchi2012chromatic}, social networks such as Facebook and Twitter (where colors indicate relationship types)~ \cite{klodt2021color,xiu2022chromatic}, and online vacation rentals based on user browsing history (where hyperedges are groups of rentals browsed during the same user session and colors indicate the location of the user)~\cite{veldt2022optimal}.

This paper addresses two limitations of previous research. First, existing
categorical clustering models do not allow for overlap between clusters. However, cluster overlap is usually necessary to accurately model the application in question. 
Researchers publish in multiple fields and often are identified with multiple areas of expertise. Many food ingredients are common across a wide range of different cuisines. Users on social media can be very active in more than one discrete time window.
The second limitation is that existing techniques do not allow for any notion of robustness to noise among the nodes.
This may be necessary to achieve accurate clusterings either
in applications where the data itself is noisy, or in applications
where the nature of the network
under consideration makes it impossible to accurately restrict the cluster memberships of
certain nodes. For example, a reasonable clustering of ingredients by cuisine should not
restrict nearly universal ingredients such as salt and flour to one or even a few clusters.
Our goal is to provide the
first clustering algorithms to simultaneously (i) respect high-order relationships (hyperedges),
(ii) provide a categorical clustering according to discrete edge labels, and (iii) contain a mechanism for allowing
overlapping clusters or robustness to noisy data points.

\textbf{Modeling overlap and robustness.} We accomplish this goal through the lens of \ECCfull{} (\ECC{}), a recent
categorical clustering framework \cite{angel2016clustering}.
Given an edge-colored hypergraph, \ECC{} asks for
assignments of colors to nodes such that, as much as possible, edges
contain nodes which match their color. Formally, a solution
makes a \emph{mistake} at edge $e$ of color $c$ if any node in $e$ is not assigned
color $c$. The goal is to produce an assignment of colors to nodes which
minimizes mistakes.
We define three new generalizations of this framework. 
These all have the same objective function as \ECC{} (minimize edge mistakes), but each comes with a parameter $b$ that defines a \emph{budget} for different ways of assigning additional colors to nodes.
\begin{itemize}[itemsep=2pt,leftmargin=*]
	\item \LOECCfull{} (\LOECC{}) seeks to minimize the number of edge mistakes, while assigning up to $b \geq 1$ colors to each node.
	\item \GOECCfull{} (\GOECC{}) allows one ``free'' color assignment per node, and additionally has $b \geq 0$ extra node-color assignments that are shared among all nodes. 
	\item \RECCfull{} also allows one ``free'' color assignment per node, and then up to $b \geq 0$ nodes are given \emph{every} color. Equivalently, up to $b$ nodes are deleted and hyperedges are contracted to contain only non-deleted nodes.
\end{itemize}
These models capture many different settings. \LOECC{} is arguably the most natural for academic co-authorship datasets, where each researcher can be associated with a small number of areas of expertise. \GOECC{} and \RECCfull{} are natural when many nodes can be associated with one category, but a few nodes should be associated with many categories (or transcend categories in the case of \RECCfull{}). For example, in clustering food ingredients, some ingredients are common in many cuisines (e.g., soy sauce or saffron) and some ingredients are so common they are in every category (e.g., salt). It is also worth noting that \GOECC{} 
is in some ways more flexible than \LOECC{}, since for a large enough budget $b$ one can always choose to give every node a small number of colors if that is the best fit for the dataset. As we shall see, this flexibility can sometimes lead to differences in theoretical results.

Each of these problems generalizes \ECC{}, with equivalence for the first
coming at $b = 1$ and for the latter two at $b = 0$.
Consequently, they are each NP-hard \cite{angel2016clustering}, motivating the study
of approximations and parameterized algorithms.

\textbf{Approximations.}
In Section \ref{sec:approximations}
we present an $r$-approximation on the number of edge mistakes for each objective, where $r$ is the
maximum hyperedge size, using a greedy approach for assigning colors. Moreover, for \LOECC{} we present a simple LP-rounding scheme
which results in a $(b+1)$-approximate solution.
In practice one might expect that flexibility should be allowed not just on the
solution size, but also on the budget used to achieve the solution. We therefore also study
bicriteria $(\alpha, \beta)$-approximations, where $\alpha$ is the approximation factor
on edge mistakes and $\beta$ is the approximation factor on the budget $b$.
For \LOECC{} and \RECC{} we present algorithms with constant $\alpha$ and $\beta$, while for
\GOECC{} we give an algorithm with constant $\beta$ and $\alpha$ linear in $b$.

\textbf{Parameterized algorithms.} In Section \ref{sec:hardness} we study the decision variants, each of which
asks whether it is possible to find a clustering which makes no more than $t$ mistakes.
We give a complete characterization of the parameterized complexity of our
problems in terms of solution size by showing that each is 
fixed parameter tractable (FPT)\footnote{A problem is FPT with respect to parameter $k$ if
	it can be solved in $f(k) \cdot n^{O(1)}$ time, where $n$ is the instance size and $f$ is any computable function.
	We refer to~\cite{flum2006parameterized} for a comprehensive study of parameterized complexity.} 
in $t+b$, but W[1]-hard in $t$ and para-NP-hard in $b$. This implies
(under standard assumptions) that they are not FPT with respect to $t$ or $b$ individually.

\textbf{Empirical analysis.} In Section \ref{sec:experiments} we present experimental results on datasets from the application domains motivated above, as well as several others. In practice, our approximation algorithms tend to significantly outperform their theoretical guarantees.
%

%% file: sections/prelims.tex
\section{Preliminaries}
Let $H = (V,E, \ell)$ denote an edge-colored hypergraph with node set $V$ and colored (hyper)edge set $E$. Throughout the text, $k$ denotes the number of colors, and $r$ is the \emph{rank} of the hypergraph, i.e., the maximum hyperedge size. The function $\ell \colon E \rightarrow [k] = \{1,2, \hdots, k\}$ labels each edge with a color, and $E_i \subseteq E$ denotes the set of edges of color $i \in [k]$. For $v \in V$, we let $E(v) \subseteq E$ denote the set of edges incident to $v$ and $d_v = |E(v)|$ be its degree. We say that color $i$ is incident to node $v$ if $v$ is contained in at least one edge in $E_i$. The {chromatic degree} $d_v^\chi$ of $v$ is the number of colors incident to $v$.

The goal of all \ECC{} problems we consider is to color nodes in a way that correlates as much as possible with edge colors, subject to different node-color constraints. To accommodate multiple color assignments, we consider maps $\lambda \colon V \rightarrow 2^{[k]}$ from nodes to the power set $2^{[k]}$ of colors, so $\lambda(v) \subseteq [k]$ represents the set of colors assigned to node $v$. Map $\lambda$ makes a \emph{mistake} at $e \in E$ if there exists a node $v \in e$ such that $\ell(e) \notin \lambda(v)$, and we say the edge is \emph{unsatisfied}. Otherwise the edge is \emph{satisfied} by $\lambda$. The set of edges where $\lambda$ makes a mistake is denoted $\mathcal{M}_\lambda \subseteq E$. Let $\mathbb{1}$ denote a binary indicator function so that $\mathbb{1}(e \in \mathcal{M}_\lambda)$ is $1$ if $\lambda$ makes a mistake at $e$ and is 0 otherwise. The optimization versions of \LOECC{}, \GOECC{}, and \RECC{} seek to minimize the number of mistakes,
\begin{equation}
	\label{eq:objective}
	{\textstyle \min_{\lambda} \sum_{e \in E} \mathbb{1}(e \in \mathcal{M}_\lambda)}
\end{equation}
subject to different constraints on the color map $\lambda$. \LOECC{} is restricted to maps satisfying $|\lambda(v)| \leq b$, \GOECC{} is restricted to maps satisfying $\sum_{v \in V} (|\lambda(v)| - 1) \leq b$ and $|\lambda(v)| \geq 1$ for every $v \in V$, and \RECC{} has the restriction that the number of nodes $v$ satisfying $\lambda(v) = [k]$ is at most $b$, while all other nodes $u \in V$ have one color assignment: $|\lambda(u)| = 1$. One can equivalently define all of these objectives as edge deletion problems, where the goal is to minimize the number of edges to delete in order to satisfy all remaining edges (subject to any constraints on $\lambda$).
One can define a maximization variant for all of these \ECC{} objectives, where the goal is to maximize the number of satisfied edges rather than minimize the number of mistakes. This is equivalent at optimality but different in terms of approximations and parameterized complexity. In this paper we focus on the mistake minimization objective.

\section{Related work}
Angel et al. \cite{angel2016clustering}
introduced the maximization variant of \ECC{} in graphs and showed that
it is NP-hard in general but polynomial-time solvable when $k = 2$. The
approximability of the maximization
variant on graphs has attracted considerable interest
\cite{ageev2015improved,ageev20200,alhamdan2019approximability}.
Amburg et al. \cite{amburg2020clustering} generalized
\ECC{} to hypergraphs, and studied the approximability of the minimization variant.
Veldt \cite{veldt2022optimal} improved the best known approximation factor for this variant to
$\min\{2(1 - \frac{1}{k}), 2(1 - \frac{1}{r+1})\}$.
Cai and Leung \cite{cai2018alternating}
showed that \ECC{} in graphs is FPT with respect
to both satisfied and unsatisfied edges. Kellerhals et al. \cite{kellerhals2023parameterized}
gave improved
parameterized algorithms for both graphs and hypergraphs, and showed
fixed-parameter (in)tractability for a variety of structural parameters.

\ECC{} is related to \textsc{Chromatic Correlation Clustering}
\cite{bonchi2015chromatic,anava2015improved,klodt2021color,xiu2022chromatic},
which generalizes \textsc{Correlation Clustering} \cite{bansal2004correlation}
to edge-colored graphs. 
Other variants of \textsc{Correlation Clustering} allow overlapping clusters~\cite{bonchi2013overlapping,andrade2014evolutionary,li2017motif}. However, none of these apply simultaneously to (i) hypergraphs, (ii) categorical interactions, and (iii) overlapping clusters.
Aboud~\cite{aboud2008correlation}
defined a \textsc{Correlation Clustering} variant in which any
node can be discarded for a price given by a penalty function
$p: V \rightarrow \mathbb{R}$. Devvrit et al. \cite{devvrit2019robust}
studied the case of a constant penalty function and added a budget for the total number of
discarded nodes, providing a bicriteria approximation.
We adopt their notion of robustness, but emphasize that
the resulting algorithms and complexity analysis are distinct, because \textsc{Correlation Clustering}
and \ECC{} have fundamentally different clustering objectives. 
More generally, notions of robustness
have been introduced in a variety of clustering frameworks
\cite{charikar2001algorithms,chen2008constant,krishnaswamy2018constant}, though to our knowledge never in combination with edge-colors.

%% file: sections/approximations.tex
\section{Approximation Algorithms}
\label{sec:approximations}
We present approximation algorithms for overlapping and robust ECC, including greedy $r$-approximations for each objective and bicriteria approximations 
based on linear programming (LP).

\subsection{Greedy algorithms}
\label{sec:greedy-approximations}
Our greedy $r$-approximation algorithms generalize the greedy algorithm for standard ECC developed by Amburg et al.~\cite{amburg2020clustering}. We first consider an alternate {linear} penalty for hyperedge mistakes, where the penalty at an edge $e \in E$ equals the number of nodes in $e$ that are not assigned the color of $e$. Formally, if $\lambda(v) \subseteq [k]$ is the set of colors assigned to $v$, this penalty is given by
\begin{equation}
	\label{eq:linpenalty}
	{\textstyle p(e,\lambda) = \sum_{v \in e} \mathbb{1}(\ell(e) \notin \lambda(v)).}
\end{equation}
For every $e \in E$ and node coloring $\lambda$, this linear penalty satisfies
\begin{equation}
	\label{eq:penaltyinequalities}
\mathbb{1}(e \in \mathcal{M}_\lambda) \leq p(e,\lambda)  \leq |e| \cdot \mathbb{1}(e \in \mathcal{M}_\lambda) \leq r \cdot \mathbb{1}(e \in \mathcal{M}_\lambda).
\end{equation}
Thus, if we can compute $\hat{\lambda} = \argmin_{\lambda} \sum_{e \in E} p(e,\lambda)$, where we minimize over a desired class of coloring functions $\lambda$ (e.g., overlapping or robust), then $\hat{\lambda}$ will be within a factor $r$ of the optimal clustering for the standard edge penalty. This applies in the same way for all overlapping and robust variants.
It remains to show that we can find this optimal $\hat{\lambda}$ in polynomial time using a greedy approach.

For node $v$ and edge $e \ni v$, if $\ell(e) \notin \lambda(v)$ we say there is a \emph{node-edge error} at $(v,e)$. Each node in an edge contributes independently to the linear edge penalty, so the linear ECC objective function simply amounts to minimizing the number of node-edge errors:
\begin{equation}
	\label{eq:linrearrange}
	 \min_\lambda \sum_{e \in E} p(e,\lambda) = \min_\lambda \sum_{v \in V} \sum_{e \ni v} \mathbb{1}(\ell(e) \notin \lambda(v)).
\end{equation}
The greedy algorithm for minimizing this objective starts with an empty labeling $\lambda$ that assigns no color to every node, and then iteratively adds node-color assignments to greedily maximize the number of node-edge errors that are fixed at each step. To formalize this, for node $v \in V$ and color $c$, let $\numvc{v}{c} = |E(v) \cap E_c|$ be the number of edges of color $c$ that $v$ belongs to. Let $\permv{v}$ be a permutation vector that arranges colors based on how often $v$ participates in them:
\begin{equation}
	\label{eq:vpreference}
	\numvc{v}{\permvi{v}{1}} \geq \numvc{v}{\permvi{v}{2}} \geq \cdots \geq \numvc{v}{\permvi{v}{k}}.
\end{equation}
We think of $\permvi{v}{i}$ as node $v$'s ``\textit{$i$-th favorite color}''. Figure~\ref{alg:greedy} gives pseudocode for the greedy procedure for \LOECC, \GOECC, and \RECC. The fact that each method optimizes the linear objective function given in~\eqref{eq:linrearrange} for each variant of overlapping or robust ECC follows from the fact that node-edge errors are independent across different nodes. If node $v$ is assigned $T$ colors, then these must be $v$'s top $T$ favorite colors, otherwise we could improve the objective by switching color assignments at that node. Furthermore, if we have a global budget of additional colors to assign (or a budget on the number of nodes that can be assigned all colors at once), it is optimal to choose nodes at each step for which additional color assignments lead to the greatest immediate improvement in the linear objective. We conclude with the following theorem.
\begin{figure}[t]

		\begin{minipage}{\linewidth}
\hrule\

			\textbf{Greedy-\LOECC}
			\begin{algorithmic}[1]
			\STATE $\lambda(v) = \{\permvi{v}{i} \colon i = 1, 2, \hdots b\}$ for each $v \in V$ 
			\end{algorithmic}
	
\vspace{2pt}
\textbf{Greedy-\GOECC}
		\begin{algorithmic}[1]
			\STATE $\lambda(v) = \{\permvi{v}{1}\}$  for each $v \in V$ 
			\FOR{$j = 1$ to $b$}
			\STATE // find $u$ whose next color choice fixes the most errors
			\STATE	$u = \argmax_v \; \numvc{v}{\permvi{v}{|\lambda(v)| + 1}}$
			\STATE // add $u$'s next favorite color to its color list
			\STATE $\lambda(u) = \lambda(u) \cup \{\permvi{v}{|\lambda(v)|+1}\}$
			\ENDFOR 
		\end{algorithmic}
\vspace{2pt}
\textbf{Greedy-\RECC}
\begin{algorithmic}[1]
	\STATE $\lambda(v) = \{\permvi{v}{1}\}$  for each $v \in V$
	\FOR{$j = 1$ to $b$}
	\STATE// find $u$ whose deletion fixes the most errors
	\STATE $u = \argmax_{v \in R} d_v - \numvc{v}{\permvi{v}{1}}$
	\STATE	$\lambda(u) = [k]$ // delete $u$ (give it all colors)
	\ENDFOR 
\end{algorithmic}

\hrule
\end{minipage}
\caption{Greedy \LOECC{}, \GOECC{}, and \RECC{}. \ The respective execution times are $O(r|E| + k|V|\log (k))$, $O(r|E| + k|V|\log{(k|V|)})$, and $O(r|E| + |V|(k + \log{(|V|)}))$.}
\label{alg:greedy}	
\end{figure}

\begin{theorem}
The greedy approach provides an $r$-approximation for the standard \LOECC{}, \GOECC{}, and \RECC{} objectives.
\end{theorem}

\subsection{Linear Programming Algorithms}
A bicriteria $(\alpha,\beta)$-approximation algorithm is one that comes within a factor $\alpha$ of the optimal number of edge mistakes while violating budget constraints by a factor at most $\beta$. Our greedy algorithms satisfy this guarantee with $\alpha = r$ and $\beta = 1$ (i.e., no budget violation), but the dependence on $r$ in the approximation factor is not ideal. We now turn to algorithms based on linear programming relaxations, which have much better approximation factors, at the expense of going over budget constraints by a small amount.

\paragraph{\LOECC{} LP Algorithms} We present the following LP relaxation for \LOECC{}:
\begin{align}
	\label{eq:loecc}
	\begin{aligned}
		\text{min} \quad& \textstyle \sum_{e \in E} x_e \\
		\text{s.t.} \quad& \forall v \in V:\\
		&\forall c \in [k], e \in E_c:\\
		&x_v^c, x_e \in [0, 1]
	\end{aligned}
	\begin{aligned}
		&\\
		&\textstyle \sum_{c=1}^k x_v^c \geq k - b \\
		\quad &x_v^c \leq x_e \quad \forall v \in e \\
		&\forall c \in [k], v \in V, e \in E.
	\end{aligned}
\end{align}
If we replace constraints $x_v^c, x_e \in [0, 1]$ with the binary constraints $x_v^c, x_e \in \{0, 1\}$, then this becomes an integer linear program (ILP) that exactly encodes \LOECC{}. The variable $x_v^c$ can be thought of as the distance from node $v$ to color $c$, and the constraint $\sum_{c} x_v^c \geq  k-b$ captures the fact that each node can be assigned at most $b$ colors. The LP relaxation can be solved in polynomial time, and its optimal solution provides a lower bound on the optimal \LOECC{} solution. We can round the fractional LP into a clustering to provide a range of bicriteria approximations that trade off in terms of budget violation and objective function approximation.
\begin{theorem}
	\label{thm:LLP}
	Let $\{x_v^c, x_e \colon e \in E, v \in V\}$ denote optimal LP variables for the LP given in~\eqref{eq:loecc}. Let $\rho \in (0,1)$ be any threshold such that $b/\rho$ is an integer, and let $\lambda$ encode a coloring where node $v \in V$ is assigned color $c \in [k]$ if $x_v^c < 1 - \rho$. This coloring is a bicriteria $\left(\frac{1}{1-\rho}, \frac{1}{\rho} - \frac{1}{b} \right)$-approximation for \LOECC{}.
\end{theorem}
\begin{proof}
	For each node $v \in V$, this rounding scheme will assign at most $b/\rho -1 = b(1/\rho - 1/b)$ colors to each node, which implies the overlap constraint is violated by at most a factor $(1/\rho - 1/b)$. Assuming this is not true implies there is some node $v \in V$ such that $x_v^c < 1-\rho$ for at least $b/\rho$ choices of the color $c$. Since $x_v^c \leq 1$ for every $c \in [k]$, we reach the following contradiction to LP feasibility:
	\begin{align*}
	{\textstyle	\sum_{c = 1}^k x_v^c < (1-\rho) \frac{b}{\rho} + k - \frac{b}{\rho} = k - b.}
	\end{align*}
	Making a mistake at $e \in E$ means some $v \in e$ was not assigned color $j = \ell(e)$, so $x_e \geq x_v^j \geq 1-\rho$. Since the cost for each edge is within a factor $1/(1-\rho)$ of the LP variable for this edge, the total cost of $\lambda$ is within a factor $1/(1-\rho)$ of the LP lower bound.
\end{proof}
It is worth noting that if we set $\rho = b/(b+1)$, this leads to a single-criteria $(b+1)$-approximation algorithm for \LOECC{} that directly generalizes the $2$-approximation for standard ECC (\LOECC{} with $b = 1$). Another significant choice of parameter is when $\rho$ is chosen to be a constant, in which case we obtain bicriteria $(\alpha,\beta)$ -approximations where both $\alpha$ and $\beta$ are constants. For example, when $\rho = 1/2$, we have a $(2, 2-1/b)$-approximation.

\paragraph{\GOECC{} LP Algorithms} For \GOECC{} we give the following LP relaxation:
\begin{align}
	\label{eq:goecc}
	\begin{aligned}
		\text{min} \quad& \textstyle \sum_{e \in E} x_e \\
		\text{s.t.} \quad& \forall v \in V:\\
		& \\
		&\forall c \in [k], e \in E_c:\\
		&x_v^c, x_e \in [0, 1]\\
		& y_v \geq 0 
	\end{aligned}
	\begin{aligned}
		&\\
		&\textstyle \sum_{c=1}^k x_v^c \geq k - y_v - 1 \\
		&\textstyle \sum_{v \in V} y_v \leq b\\
		\quad &x_v^c \leq x_e \quad \forall v \in e \\
		&\forall c \in [k], v \in V, e \in E\\
		&\forall v \in V.
	\end{aligned}
\end{align}
To see why this is a relaxation, note that if $\{x_v^c, x_e\}$ variables were constrained to be binary, we would recover the ILP for \GOECC{}, where the variable $y_v$ is an integer encoding the number of \emph{additional} colors assigned to node $v$ beyond its first color assignment. Solving the relaxation produces nonnegative $y_v$ values that are generally not integers. For our LP rounding procedure, we define the following function for $x \geq 0$ and threshold $\delta \in (0,1)$:
\begin{equation}
	\label{eq:roundx}
	\lfloor x \rceil_\delta  = \begin{cases}
		\lfloor x \rfloor & \text{if $x - \lfloor x \rfloor < \delta$} \\
		\lceil x \rceil & \text{ otherwise}.
	\end{cases}
\end{equation} 
The following useful properties hold for every $x \geq 0$:
\begin{align}
	\label{eq:prop1}
	\lfloor x \rceil_\delta + \delta \geq x, \\
	\label{eq:prop2}
	\lfloor x \rceil_\delta \leq x + (1-\delta), \\
	\label{eq:prop3}
	\lfloor x \rceil_\delta \leq \frac{1}{\delta} x.
\end{align}

\begin{theorem}
	\label{thm:GLP}
	Let $\{x_v^c, x_e, y_v \colon e \in E, v \in V\}$ denote optimal LP variables for the LP given in~\eqref{eq:goecc}. For a threshold $\delta \in (0,1)$, let $\rho_v = \frac{1-\delta }{\lfloor y_v \rceil_\delta +2}$ for each $v \in V$, and let $\lambda$ encode a coloring where node $v \in V$ is assigned color $c \in [k]$ if $x_v^c < \rho_v$. This coloring is a bicriteria $\left(\frac{(b+2)}{1-\delta} + 1, \frac{1}{\delta} \right)$-approximation for \GOECC{}.
\end{theorem}
\begin{proof}
	This rounding scheme will assign at most $\lfloor y_v \rceil_\delta + 1$ colors to node $v$. If not, then $x_v^c < \rho_v$ for at least $\lfloor y_v \rceil_\delta + 2$ choices of the color $c$, so using the inequality in~\eqref{eq:prop1} we reach a contradiction to the LP constraints:
	\begin{align*}
		\sum_{c = 1}^k x_v^c &< (\lfloor y_v \rceil_\delta + 2)\rho_v + k - \lfloor y_v \rceil_\delta - 2 \\
		&= (1-\delta) + k - \lfloor y_v \rceil_\delta - 2 \\
		&= k - (\lfloor y_v \rceil_\delta + \delta) - 1 \leq k - y_v - 1.
	\end{align*}
	Therefore, since each node is assigned one color without using any of the budget $b$, property~\eqref{eq:prop3} shows this rounding scheme assigns an extra $\sum_{v} \lfloor y_v \rceil_\delta \leq \frac{1}{\delta} \sum_{v} y_v \leq \frac{1}{\delta} b$ color assignments, violating the budget by at most a factor $1/\delta$.
	
	If we make mistake at $e \in E$, then $x_e \geq x_v^{\ell(e)} \geq \rho_v$ for some $v \in V$, so this mistake is within a factor
	\begin{align*}
		\frac{1}{\rho_v} = \frac{ \lfloor y_v \rceil_\delta + 2}{1-\delta} \leq \frac{ y_v + (1-\delta) + 2}{1-\delta} \leq \frac{(b+2)}{1-\delta} + 1
	\end{align*}
	of the LP lower bound for that edge.
\end{proof}
By setting $\delta = \frac{1}{2}$ in the above theorem, we obtain a bicriteria $(2b+5,2)$-approximation. We can also come arbitrarily close to satisfying budget constraints as long as we are willing to make the approximation to the objective function worse. In general, the fact that the overlap budget is shared by all nodes makes it more challenging to obtain good approximations for \GOECC{}. Unlike our LP algorithms for \LOECC{}, there is no choice of threshold $\delta$ that leads to a single-criteria approximation algorithm (i.e., no violation of the budget), nor a setting that produces a constant-constant bicriteria approximation for \GOECC{}. Exploring alternate rounding techniques and relaxations that achieve these goals is therefore a natural direction for future research.

\paragraph{\RECC{} LP Algorithms} Finally, we present the following LP relaxation for \RECC{}:
\begin{align}
	\label{eq:recc}
	\begin{aligned}
		\text{min} \quad& \textstyle  \sum_{e \in E} x_e \\
		\text{s.t.} \quad& \forall v \in V:\\
		&\forall c \in [k], e \in E_c:\\
		&\\
		&x_v^c, x_e, z_v \in [0, 1]\\
		& y_v \geq 0 
	\end{aligned}
	\begin{aligned}
		&\\
		&\textstyle \sum_{c=1}^k x_v^c \geq k - 1 \\
		\quad &x_v^c - z_v \leq x_e \quad \forall  v \in e \\
		&\textstyle \sum_{v \in V} z_v \leq b\\
		&\forall  c \in [k], v \in V, e \in E\\
		&\forall  v \in V.
	\end{aligned}
\end{align}
Similar to our LPs for overlapping ECC, constraining variables $\{x_v^c, x_e, z_v\}$ to be binary produces an ILP that exactly encodes \RECC{}. The variable $z_v$, when binary, encodes whether node $v$ should be deleted ($z_v = 1$) or not ($z_v = 0$). If we delete this node, then every edge $e \ni v$ is not violated even if $v$ is not assigned color $\ell(e)$, which is captured by the constraint $x_v^c - z_v \leq x_e$.
\begin{theorem}
	\label{thm:RLP}
	Let $\{x_v^c, x_e, z_v \colon e \in E, v \in V\}$ denote optimal LP variables for the LP in~\eqref{eq:recc}. For $\varepsilon \in (0,\frac12)$, let $\lambda$ encode a coloring where $v \in V$ is deleted if $z_v \geq \varepsilon$ and is given color $c \in [k]$ if $x_v^c < \frac12$. Then $\lambda$ is a bicriteria $\left(\frac{2}{1-2\varepsilon}, \frac{1}{\varepsilon} \right)$-approximation for \RECC{}.
\end{theorem}
\begin{proof}
	The fact that $v$ is given color $c$ only if $x_v^c < \frac12$ implies that each node will be assigned at most one color. If $\lambda$ makes a mistake at an edge $e \in E$, this means that there is some $v \in e$ that is not deleted and is not given color $\ell(e)$, so
	\begin{align*}
		x_e \geq x_v^{\ell(e)} - z_v > \frac{1}{2} - \varepsilon = \frac{1- 2\varepsilon}{2}.
	\end{align*}
	Thus, we can pay for all edge mistakes within a factor $2/(1-2\varepsilon)$ of the LP lower bound. Define $\hat{z}_v = 1$ if we delete node $v$ (i.e., $z_v \geq \varepsilon)$ and let $\hat{z}_v = 0$ otherwise. Observe that $\hat{z}_v \leq (1/\varepsilon) z_v$, so the number of deleted nodes is $\sum_{v} \hat{z}_v \leq \frac{1}{\varepsilon} \sum_{v} z_v \leq \frac{1}{\varepsilon} b$ which proves the desired bound on the budget violation.
\end{proof}
If we set $\varepsilon$ to a small constant, we can get constant-constant bicriteria approximation guarantees. For example, $\varepsilon = 1/3$ produces a $(6,3)$-approximation and $\varepsilon = 1/4$ gives a $(4,4)$-approximation. 
It is worth noting that it is actually impossible to obtain a single-criteria approximation for \RECC{} by rounding the LP relaxation. To see why, consider a hypergraph $H$ with four nodes $V = \{v_1,v_2,v_3,v_4\}$ and two hyperedges $e_1 = \{v_1,v_2,v_3\}$ and $e_2 = \{v_2,v_3,v_4\}$ with unique colors $1$ and $2$. When $b = 1$, the optimal \RECC{} solution makes one mistake, but the LP relaxation obtains an objective score of 0 by setting $x_{v_2}^1 = x_{v_2}^2 = x_{v_3}^1 = x_{v_3}^2 = z_{v_2} = z_{v_3} = 1/2$. Because every labeling that respects the budget will make at least one mistake, the ratio between the number of mistakes and the LP lower bound is infinite. 

%% file: sections/hardness.tex
\section{Parameterized Complexity}
\label{sec:hardness}

Our parameterized complexity results focus on the decision versions of \LOECC{}, \GOECC{}, and \RECC{}.
The input in these cases is an edge-colored hypergraph $H = (V,E,\ell)$ and two non-negative integers $b$ and $t$.
\LOECC{} asks whether there exists a subset $X \subseteq E$ of size $|X| = t$ such that every edge in $E \setminus X$ can be satisfied by assigning at most $b$ colors to each node.
Decision versions for \GOECC{} and \RECC{} are defined similarly using their corresponding constraints on $\lambda$.
We use a tuple $(H,t,b)$ to denote an instance of a decision problem.

All three problems are XP in $t$ (i.e., admit an $O(|H|^{f(t)})$ algorithm).
One such algorithm exhaustively tries every way to remove $t$ edges and checks if coloring (or removing for \RECC{}) the nodes to satisfy the remaining edges violates the budget.

\subsection{W[1]-Hardness}

\LOECC{}, \GOECC{}, and \RECC{} are all para-NP-hard (implying W[1]-hardness) with respect to $b$, since standard \ECC{} is already NP-hard (i.e., $b = 0$ for \GOECC{} and \RECC{}, and $b = 1$ for \LOECC{}).
This connection also rules out FPT algorithms in various other
potentially interesting parameters, including
$k$, $r$, max degree, vertex-cover number, and treewidth~\cite{kellerhals2023parameterized}. Here, we show that each of our problems is
W[1]-hard in the natural parameter $t$. For \LOECC{}, we give a reduction from \SCov{}, which is
W[2]-hard
with respect to the size
of the cover~\cite{flum2006parameterized}. 
W[2]-hardness implies W[1]-hardness by definition~\cite{flum2006parameterized}.

\problembox{Set Cover}
{a set of elements $\mcu = \{x_1, \dots, x_n\}$, a family of sets $\mcf = \{S_1, \dots, S_m\}$, and an integer $t \in \mathbb{N}$}
{is there a subset $Y \subseteq \mcf$ such that $|Y| \leq t$ and $\bigcup_{S \in Y} S = \mcu$?}



\begin{theorem}
\label{thm:loecc-w2}
    \LOECC{} is W[2]-hard with respect to $t$.
\end{theorem}

\begin{proof}
    Given an instance $(\mcu, \mcf, t)$ of \SCov{}, let $(H, t, b=m-1)$ be the following instance of \LOECC{}.
    For each element $x_j \in \mcu$, add the node $u_j$ to $H$.
    For each set $S_i \in \mcf$, add the node $v_i$ to $H$.
    Additionally, add the hyperedge $e_i$ with color $i$ to $H$ containing $v_i$ and the vertices $\{u_j : x_j \in S_i\}$.
    Finally, for each node in $\{u_j : x_j \not\in S_i\}$, add a hyperedge containing $u_j$ and $v_i$ with color $i$.
    Note that each node $u_j$ participates in $m$ hyperedges with $m$ different colors.
    We argue that $(H, t, b)$ is a yes-instance of \LOECC{} if and only if $(\mcu, \mcf, k)$ is a yes-instance of \SCov{}.

    First, assume that $(\mcu, \mcf, t)$ is a yes-instance of \SCov{}, and let $Y \subseteq \mcf$ be a feasible solution.
    We claim that $X = \{e_i : S_i \in Y\}$ is a feasible solution for $(H, t, b)$.
    By construction, $|X| \leq t$.
    Consider whether the vertices of $H$ can be assigned colors to satisfy the remaining hyperedges.
    Each node $v_i$ is only contained in hyperedges of color $i$ and thus only needs one color.
    Since $Y$ is a feasible set cover, each node $u_j$ must be contained by a hyperedge in $X$.
    Thus, $u_j$ participates in at most $m-1$ hyperedges and only needs $b=m-1$ colors.
    As a result, $X$ is a feasible solution, and $(H, t, b)$ is a yes-instance of \LOECC{}.

    Now, assume that $(H, t, b)$ is a yes-instance of \LOECC{}, and let $X$ be a feasible solution.
    Let $Y = \{S_i : e_i \in X\}$.
    For the remaining hyperedges in $X$ of the form $\{v_i, u_j\}$, add any $e_{i'}$ to $Y$ such that $x_j \in S_{i'}$.
    By construction, $|Y| \leq t$.
    Suppose that $Y$ is not a set cover, and there exists an uncovered element $x_j \in \mcu$.
    By our choice of $Y$, no hyperedges containing $u_j$ appear in $X$.
    However, since $u_j$ participates in $m$ hyperedges with $m$ different colors, assigning only $b$ colors to $u_j$ cannot satisfy all incident hyperedges.
    This contradicts that $X$ was a feasible solution, and so $Y$ must also be feasible.
    Thus, $(\mcu, \mcf, k)$ is a yes-instance of \SCov{}.
\end{proof}

To prove W[1]-hardness for \GOECC{} and \RECC{}, we give a reduction from \PVC{},
which is W[1]-hard with respect to the size of the cover~\cite{guo2007parameterized}.

\problembox{Partial Vertex Cover}
{a graph $G$ and non-negative integers $t$ and $s$}
{is there a set $Y$ of at most $t$ vertices which cover at least $s$ edges?}

\begin{theorem} \label{thm:goecc-w1}
    \GOECC{} is W[1]-hard with respect to $t$.
\end{theorem}

\begin{proof}
    Given an instance $(G, t, s)$ of \PVC{} with $m$ edges, construct a \GOECC{} instance $(H, t, b=m-s)$ by mapping edges to nodes and vertices to hyperedges.
    For each edge $f \in G$, add the node $u_f$ to $H$.
    For each vertex $v \in G$, add the hyperedge $e_v$ containing the nodes $u_f$ corresponding to edges $f$ incident to $v$ in $G$.
    Each hyperedge is given a unique color.
    Note that each node $u_f$ participates in exactly two hyperedges.
    We argue that $(H, t, b)$ is a yes-instance of \GOECC{} if and only if $(G, t, s)$ is a yes-instance of \PVC{}.

    First, assume $(G, t, s)$ is a yes-instance of \PVC{}, and let $Y$ be a feasible solution.
    We claim that $X = \{e_v : v \in Y\}$ is a feasible solution for $(H, t, b)$.
    By construction, $|X| \leq t$.
    In $H \setminus X$, a node $u_f$ only needs more than one color if the edge $f$ is uncovered by $Y$ in $G$.
    Since $Y$ covers at least $s$ edges, there are at most $b=m-s$ of these nodes.
    Moreover, since each node needs at most two colors, we need at most $b$ additional color assignments to satisfy every hyperedge.
    Therefore, $X$ is a feasible solution, and $(H, t, b)$ is a yes-instance of \GOECC.

    Now, assume $(H, t, b)$ is a yes-instance of \GOECC{}, and let $X$ be a feasible solution.
    Let $Y = \{v : e_v \in X\}$.
    By construction, $|Y| \leq t$.
    Suppose that $Y$ is not a feasible partial vertex cover, and there exists a set $U$ of $m-s+1$ uncovered edges in $G$.
    For every edge $f = (v, w) \in U$, neither $e_v$ nor $e_w$ are in $X$.
    Thus, $u_f$ appears in hyperedges with two different colors, and an additional color assignment is necessary to satisfy both $e_v$ and $e_w$.
    However, then we must use a global budget of at least $m-s+1 > b$ colors, contradicting that $X$ is a feasible solution.
    Therefore, $Y$ covers at least $s$ edges, and $(G, t, s)$ is a yes-instance of \PVC{}.
\end{proof}

\begin{theorem}
\label{thm:recc-w1}
   \RECC{} is W[1]-hard with respect to $t$.
\end{theorem}

The proof proceeds identically to Theorem~\ref{thm:goecc-w1} since removing a node has the same effect as assigning a second color.

\subsection{FPT Algorithms}

Since \LOECC{}, \GOECC{}, and \RECC{} are not FPT with respect to either $t$ or $b$ individually, we now consider parameterizing by $t+b$.
We give branching algorithms which work by defining and resolving \emph{conflicts}: nodes which cannot satisfy all of their incident hyperedges without deleting a hyperedge or using an extra color.

\begin{theorem}
    \LOECC{} is FPT with respect to $t+b$.
\end{theorem}

\begin{proof}
    Given a \LOECC{} instance $(H, t, b)$, a \emph{conflict} is a node $v$ and $b+1$ incident hyperedges each with a unique color.
    If there are no conflicts, all edges can be satisfied since each node needs at most $b$ colors to satisfy all incident hyperedges.
    If $(H, t, b)$ contains a conflict, we branch on the $b+1$ possible deletions to resolve it.
    Each branch increases the number of deleted edges by one, so the search tree has depth at most $t$.
    Conflicts can be found in $O(r|E|)$ time by checking the set of incident hyperedges at each node.
    Thus, the algorithm runs in $O((b+1)^t r|E|)$ time which is FPT in $t+b$.
\end{proof}

\begin{theorem} \label{thm:goecc-fpt}
    \GOECC{} is FPT with respect to $t+b$.
\end{theorem}

\begin{proof}
    Given an instance $(H, t, b)$ of \GOECC{}, recall that $\lambda(v)$ denotes the set of colors assigned to a node $v$.
    We define a \emph{conflict} to be a node $v$ and the incident hyperedges $e$ of color $c_e$ and $f$ of color $c_f$ such that $c_e \neq c_f$ and $c_e, c_f \not\in \lambda(v)$.
    First, we note that in an instance $(H, t, b)$ partially colored by $\lambda$ with no conflicts, any unsatisfied hyperedges must be monochromatic at each node.
    Thus, if $\lambda$ assigns at most $b$ colors, we can satisfy the remaining edges in $H$ in linear time by assigning one (free) color to each node.

    If $(H, t, b)$ does contain a conflict, we branch on the four possible ways to resolve it: assign the color $c_e$ to $v$, assign the color $c_f$ to $v$, delete $e$, or delete $f$.
    All four branches decrease either $t$ or $b$, and so the maximum depth of the search tree is at most $t+b$.
    Moreover, we can find conflicts in $O(r|E|)$ time by checking the set of incident hyperedges at each node.
    Thus, the algorithm runs in $O(4^{t+b} r|E|)$ time which is FPT in $t+b$.
\end{proof}

\begin{theorem}
\label{thm:recc-fpt}
    \RECC{} is FPT with respect to $t+b$.
\end{theorem}

The FPT algorithm for \RECC{} proceeds in the same manner as \GOECC{} in Theorem~\ref{thm:goecc-fpt}, the main difference being that a node deletion potentially resolves more conflicts than assigning a single additional color.

\subsection{Kernelization}

Finally, we observe that a single reduction rule leads to a kernel for each variant.
Specifically, we need only remove \emph{easy nodes}.
A node $v$ is easy if it can be colored without conflict (i.e.,\ only in hyperedges of one color for \GOECC{}/\RECC{} or at most $b$ colors for \LOECC{}).
Note that removing easy vertices may result in hyperedges with only one node.

\begin{theorem}
\label{thm:kernels}
    \LOECC{} admits a kernel with $rt$ vertices.
    \GOECC{} and \RECC{} admit a kernel with $rt+b$ vertices.
\end{theorem}

\begin{proof}
    First, remove all easy nodes.
    Since deleting a hyperedge of size $r$ resolves conflicts in at most $r$ nodes, deleting $t$ edges resolves conflicts in at most $rt$ nodes.
    In \GOECC{} and \RECC{}, conflicts in an additional $b$ nodes can be resolved by using the global budget to assign an additional color or delete a node.
    Since we must resolve at least one conflict in every node after removing easy nodes, yes-instances have a bounded size.
\end{proof}

From here, multiple brute force algorithms yield FPT results of varying quality depending on the number of colors, the maximum hyperedge size $r$, and the value of $t+b$.

%% file: sections/experiments.tex
\section{Experiments}
\label{sec:experiments}

\begin{table*}[t]
    \centering
    \caption{Summary statistics of datasets - number of nodes $|V|$; number of (hyper)edges
    $|E|$; number of edge colors $k$; maximum and mean hyperedge size $r$
    and $\mu(|e|)$; and maximum and mean chromatic degree $\Delta^{\chi}$
    and $\mu(d^{\chi})$. Also, categorical edge clustering performance for
    the algorithms \LLPfull{} (\LLP{}), \LGfull{} (\LG{}), \GLPfull{} (\GLP{}),
    \GGfull{} (\GG{}), \RLPfull{} (\RLP{}), and \RGfull{} (\RG{}). \LLP{} was run
    with local budgets $b \in \{1, 2, 3, 4, 5, 8, 16, 32\}$. \GLP{} was run
    with budgets $b$ such that
    $b/|V| \in \{0, 0.5, 1, 1.5, 2, 2.5, 3, 3.5, 4\}$. 
    \RLP{} was run 
    with budgets $b$ such that $b/|V| \in 
    \{0, .01, .05, .1, .15, .2, .25\}$. \looseness=-1 Performance is listed in terms of the approximation
    guarantee given by the LP lower bound (lower is better) and each listed
    value is the maximum (worst) for the given algorithm across all tested values of $b$, rounded up to three
    decimal places.
    }
    \begin{tabular}{l l l l l l l l l l l l l l l l l}
        \hline
        \multicolumn{8}{c}{} & \multicolumn{6}{c}{Maximum $\alpha$} & \multicolumn{3}{c}{Maximum $\beta$}\\
        \cmidrule(lr){9-14} \cmidrule(lr){15-17}
        \emph{Dataset} & $|V|$ & $|E|$ & $k$ & $r$ & $\mu (|e|)$ & $\Delta^{\chi}$ & $\mu(d^{\chi})$ & \LLP{} & \LG{} & \GLP{} & \GG{} & \RLP{} & \RG{} & \LLP{} & \GLP{} & \RLP{} \\
        \hline
        \BRAIN{}          &638    &21180  &2 &2 &2 &2 &1.92 &1 &1.013 &1 &1.016 &1.007 &$2.785^{\ast}$ &1 &1 &2\\
        \MAG10{}         &80198 &51889 &10 &25 &3 &9 &1.26 &1.002 &1.2 &1 &1.2 &1.098 &$\infty^{\ast}$ &1 &1 &1.804\\
        \COOKING{}        &6714 &39774 &20 &65 &12 &20 &4.35 &1 &1.858 &1.010 &2.867 &1.149 &1.471 &1 &1 &1.678\\
        \DAWN{}           &2109 &87104 &10 &22 &4 &10 &3.72 &1.001 &1.614 &1.004 &2.427 &1.135 &1.670 &1 &1 &1.477\\
        \WALMART{}        &88837 &65898 &44 &25 &5 &40 &2.65 &1.002 &1.432 &1.001 &9.103 &1.242 &4.334 &1 &1 &1.491\\
        \TRIVAGO{}        &207974 &247362 &55 &85 &2 &32 & 1.55 &1.005 &1.347 &1.002 &9.113 &1.048 &$\infty^{\ast}$ &1.5 &1 &1.581\\
        \hline
    \end{tabular}
    \label{tab:summary}
    \parbox[t]{\textwidth}{$^{\ast}$ \footnotesize \textit{Note that observed $\alpha$ values can exceed $r$ because the $LP$ provides a
    \emph{lower bound} on the optimal number of edge mistakes. In the \MAG10{} and \TRIVAGO{} cases, the LP found fractional solutions
    with objective score 0. \RLP{} was able to create a rounded solution with no edge mistakes by exceeding the node deletion
    budget $b$, but \RG{} was not.}}
\end{table*}

We evaluate the approximation algorithms developed in Section \ref{sec:approximations} on 
a corpus of six real-world datasets. These have previously served as benchmarks
for \ECCfull{} algorithms \cite{amburg2020clustering, veldt2022optimal}, and capture many motivating settings for overlapping categorical clustering (e.g., co-authorship datasets and food ingredient datasets). Though we do not have optimal cluster assignments,
we are able to upper bound the approximation ratios via comparison to the optimal LP
objective scores.
Because
our algorithms are the first of their kind, there are no competitors against which to compare.
The exception is for choices of $b$ which recover non-overlapping \ECC{}, in which cases we observe consistency with
prior benchmarking results~\cite{amburg2020clustering,veldt2022optimal}.
Otherwise, these
experiments demonstrate that our approximation algorithms
can produce optimal or near-optimal solutions to their (NP-hard) objectives, oftentimes with
amazingly low runtimes.
We then discuss several notable experimental outcomes, building intuition for
the ways in which (hyper)graph structure affects \ECCfull{} algorithm performance. Finally, we
compare our two models for \OECCfull{}. Our code and datasets are available at
\href{https://github.com/TheoryInPractice/overlapping-ecc}{https://github.com/TheoryInPractice/overlapping-ecc}.

\textbf{Datasets.} Our datasets are summarized in Table \ref{tab:summary}.
\BRAIN{} \cite{crossley2013cognitive} is a graph in which nodes represent brain
regions and edges represent relationships revealed by MRI scans. There are two edge
colors -- one for pairs of regions with high fMRI correlation and another for
pairs with similar activation patterns. The \textsf{Drug Abuse Warning Network} (\DAWN{})
\cite{dawn2011} models drugs with nodes and patients with hyperedges indicating the
combination of drugs taken prior to an emergency room visit. Edges are colored by the 
patient's emergency room results (e.g., ``sent home'' or ``surgery''). The \MAG10{} hypergraph is constructed from
the \textsf{Microsoft Academic Graph} \cite{sinha2015overview}, with nodes representing authors, hyperedges
capturing co-authorship groups, and edge colors representing computer science venues
to which the groups have submitted papers. When the same group of co-authors has submitted
to multiple venues, the most common color is chosen, and ties are discarded. The
\COOKING{} dataset \cite{kaggle2015cooking} represents ingredients with nodes and recipes with
hyperedges. Edge colors are used to signify cuisine type. In the
\WALMART{} dataset, nodes represent individual products, hyperedges
capture groups of products that have been purchased together in a single shopping
trip, and edge colors are categorical ``trip type'' labels assigned by
Walmart \cite{kaggle2015walmart}. Finally, in the \TRIVAGO{} dataset~\cite{trivagodata}
nodes are
vacation rental properties on the titular website, and edges represent groups of
properties viewed by a single user during a single browsing session. Edge colors track
the country from which the browsing session took place. Thus, the goal in the context of
\ECC{} is to use the resulting hypergraph to cluster properties according to the countries
from which they are likely to attract renters.

\begin{figure}%
    \begin{subfigure}[t]{.49\linewidth}
        \includegraphics[width=\linewidth]{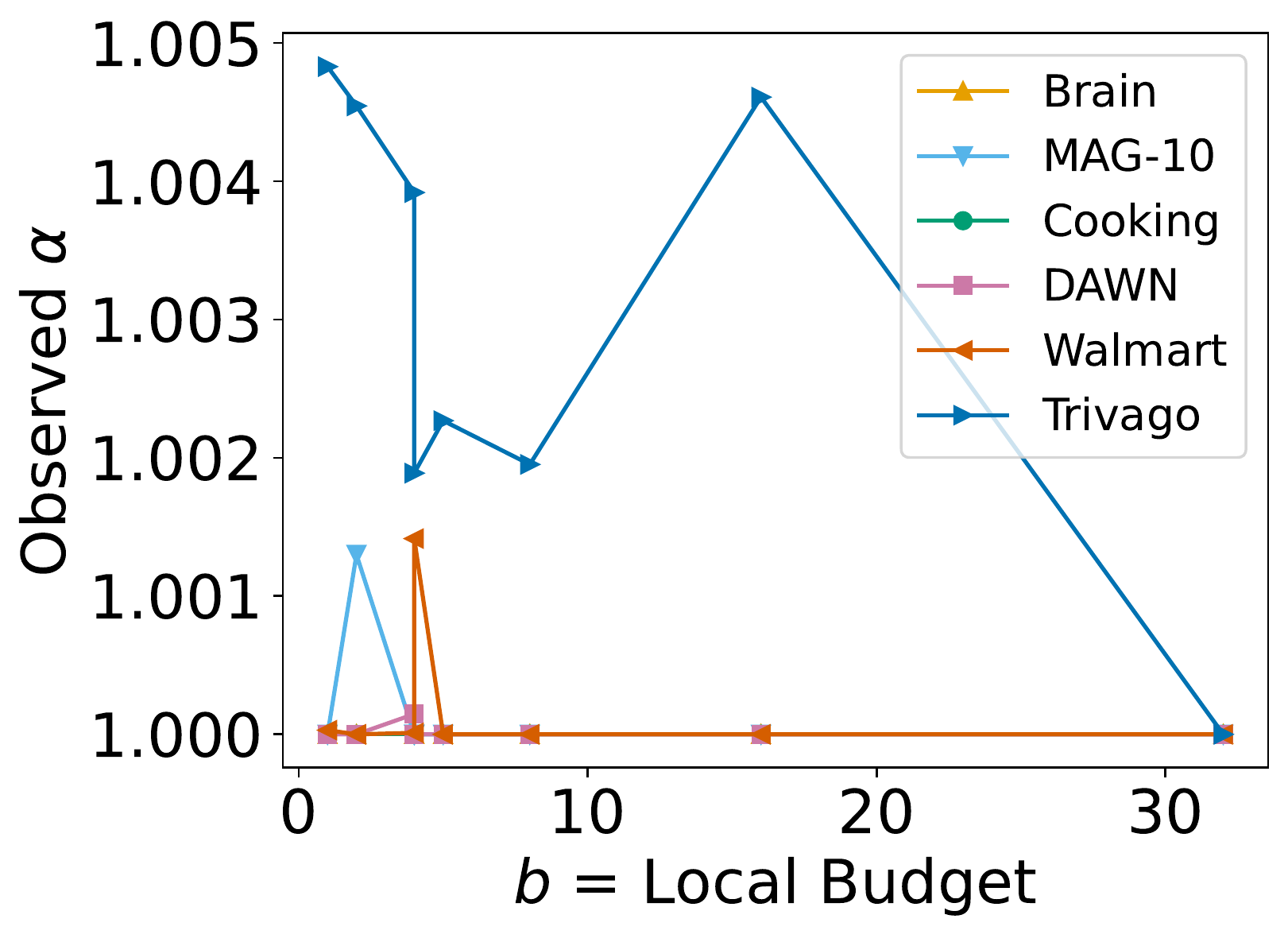}
        \caption{\LLP{} Observed $\alpha$}
        \label{fig:lo-alphas}
    \end{subfigure}\hspace{\fill} 
    \begin{subfigure}[t]{.49\linewidth}
        \includegraphics[width=\linewidth]{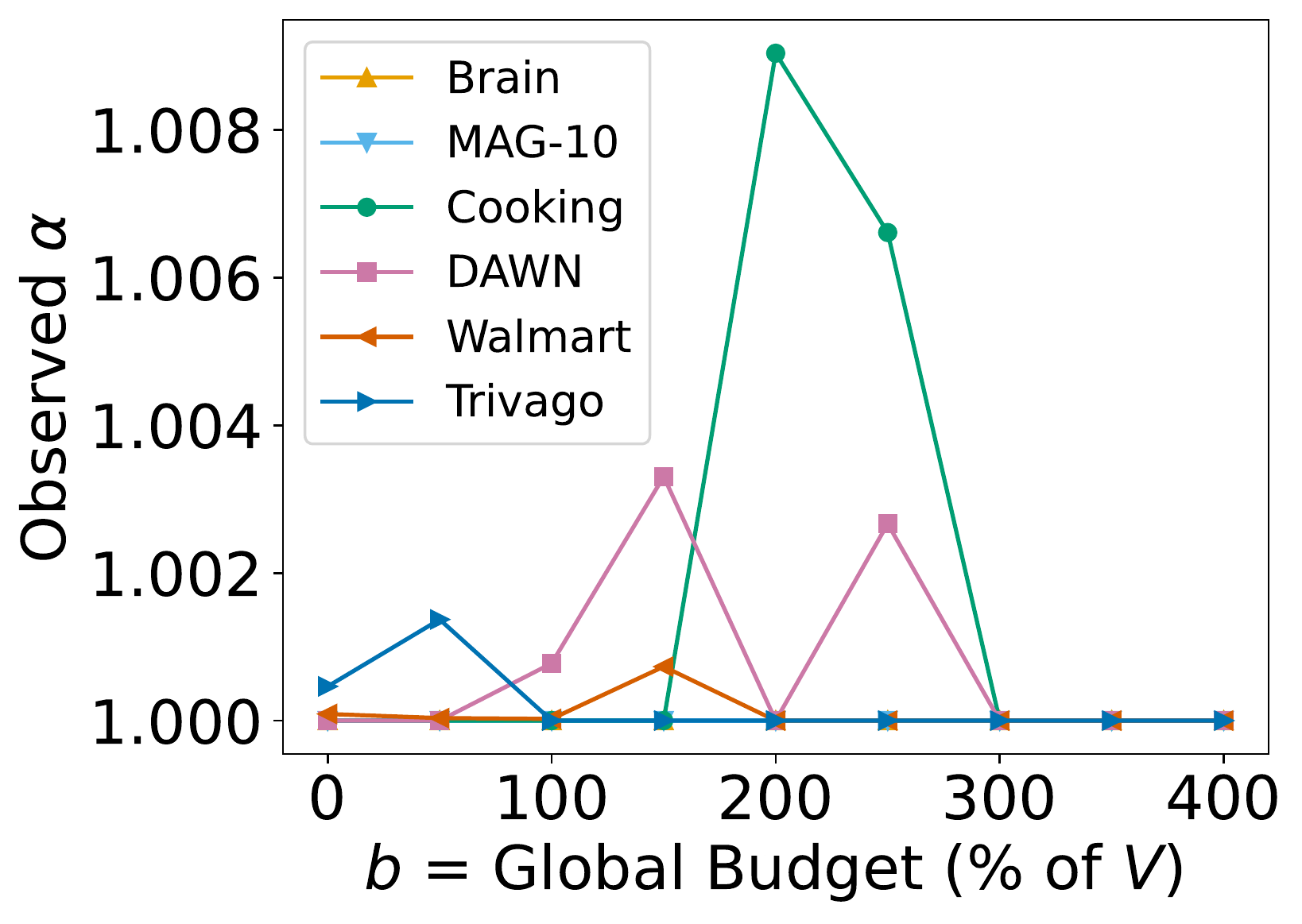}
        \caption{\GLP{} Observed $\alpha$}
        \label{fig:go-alphas}
    \end{subfigure}
    \bigskip
    \begin{subfigure}[t]{.49\linewidth}
        \includegraphics[width=\linewidth]{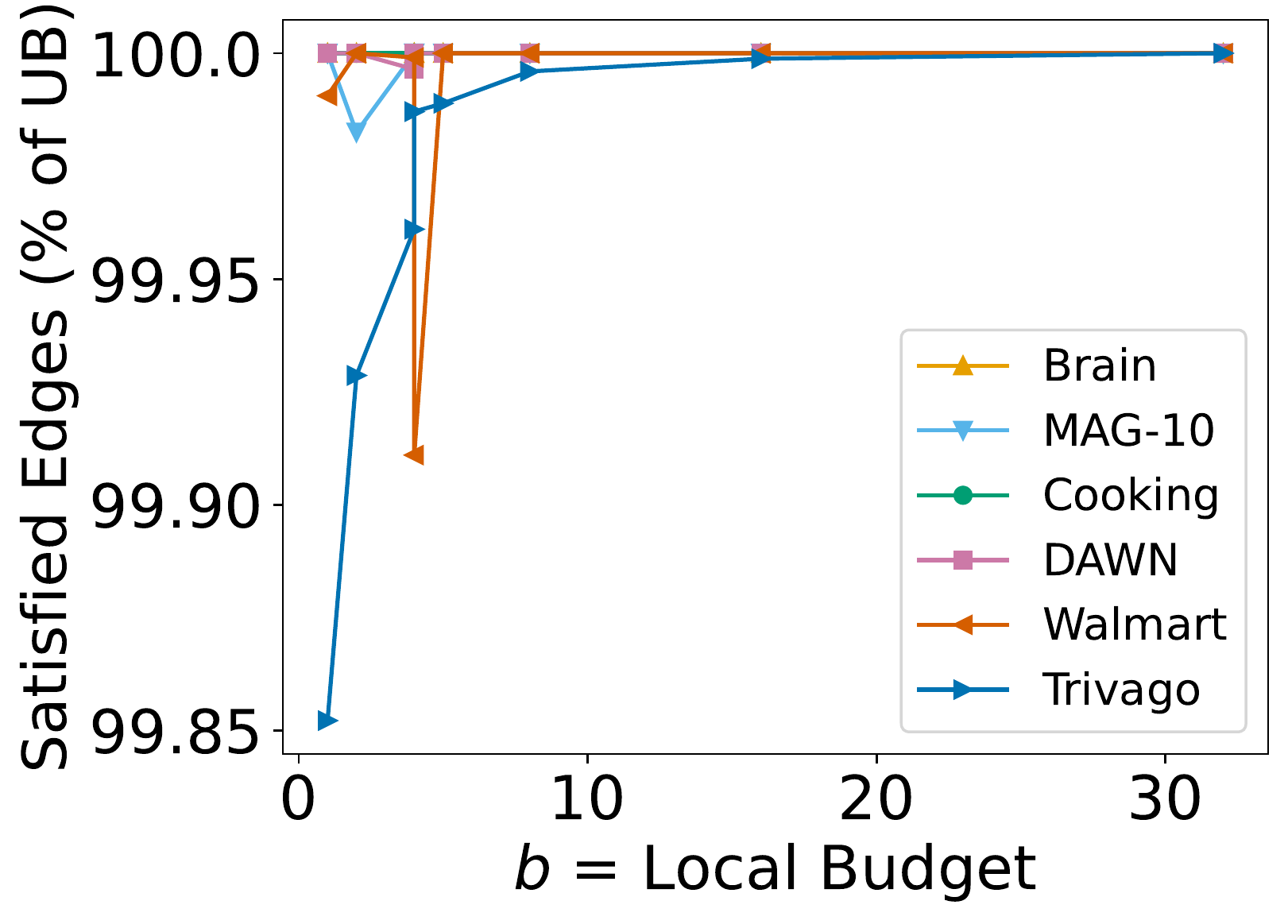}
        \caption{\LLP{} Satisfied Edges}
        \label{fig:lo-satisfied-percent}
    \end{subfigure}\hspace{\fill} 
    \begin{subfigure}[t]{.49\linewidth}
        \includegraphics[width=\linewidth]{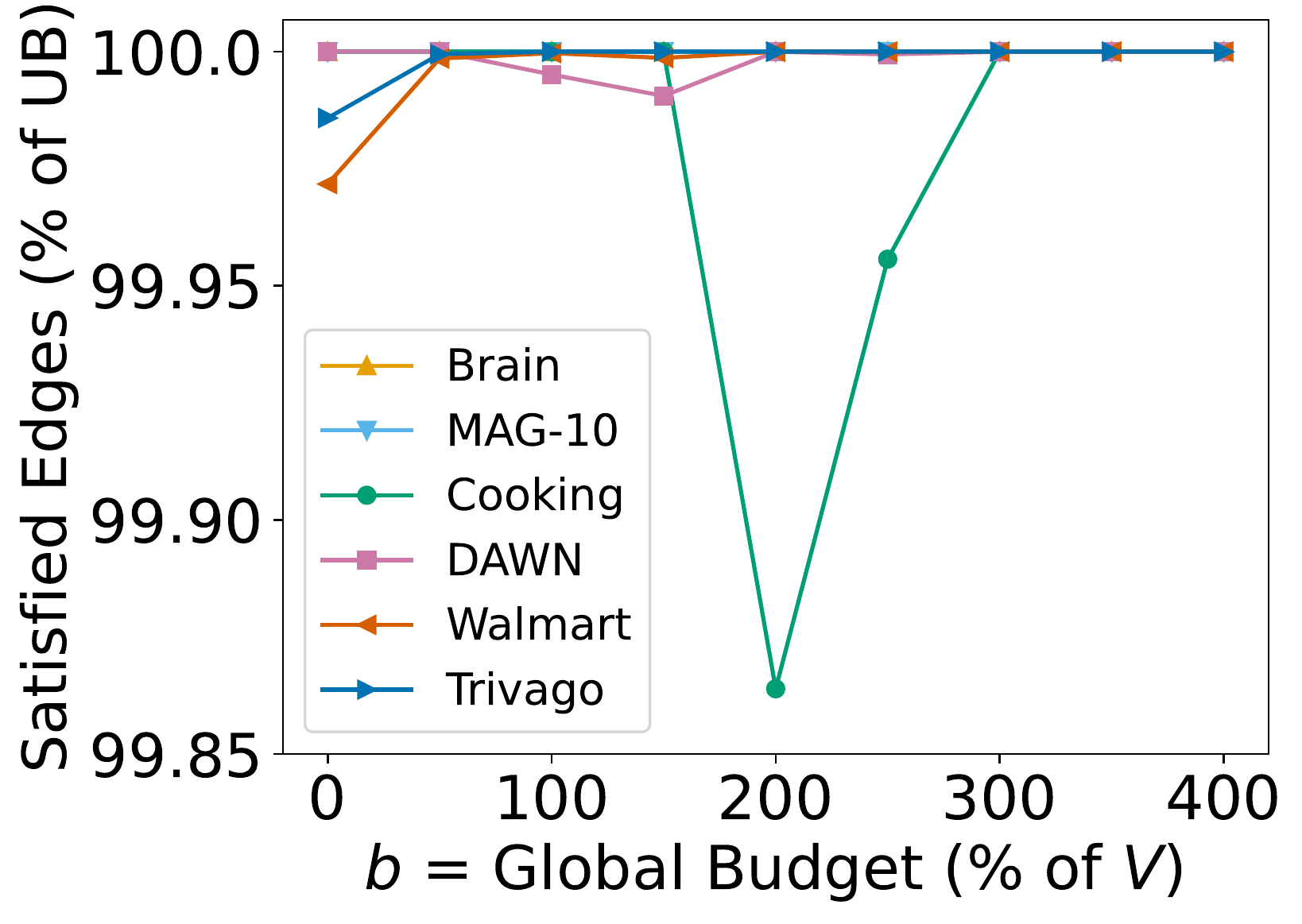}
        \caption{\GLP{} Satisfied Edges}
        \label{fig:go-satisfied-percent}
    \end{subfigure}
    \caption{(a)-(b): Observed $\alpha$ values for \LLP{} and \GLP{}, where the performance is evaluated
    against the LP lower bound. (c)-(d): Satisfied edge set sizes for
    \LLP{} and \GLP{}, presented as percentages of the upper bound derived from the LP.}%
    \label{fig:logo-alphas}%
\end{figure}

\vspace{-.2cm}
\subsection{Algorithm Performance}

We evaluate each of the six algorithms developed in Section \ref{sec:approximations},
with particular emphasis on the bicriteria LP-rounding algorithms described by
Theorems
\ref{thm:LLP}, \ref{thm:GLP}, and \ref{thm:RLP}. 
We refer to these as \LLPfull{} (\LLP{}), \GLPfull{} (\GLP{}), and
\RLPfull{} (\RLP{}), respectively. We select parameters to produce small approximation guarantees on
both the edge mistake and budget objectives (constant-constant where possible). Specifically, 
we test \LLP{} with $\rho = 1/2$, \GLP{} with
$\delta = 1/2$, and \RLP{} with $\varepsilon = 1/3$, guaranteeing approximation factors
$(\alpha = 2, \beta = 2 - 1/b)$, $(2b + 5, 2)$, and $(6, 3)$. We test with a variety
of values for the budget $b$, in each case selected to be representative of the practically
useful parameter space. 

Table \ref{tab:summary}
reports how well each algorithm approximates its objective. The maximum $\alpha$
is the ratio between the number of mistakes made by a given algorithm and the lower bound
provided by the LP relaxation, maximized across all tested budgets $b$. The maximum
$\beta$ is the approximation factor on the budget, once again
maximized across all tested values of $b$. Amazingly, we find that all three LP-rounding
algorithms drastically outperform their theoretical guarantees. In particular,
Figures \ref{fig:logo-alphas}\subref{fig:lo-alphas} and \ref{fig:logo-alphas}\subref{fig:go-alphas} show that both
\OECCfull{} algorithms provide nearly perfect performance on the edge mistake objective.
Indeed, both LPs frequently find optimal (integral) solutions, leaving no rounding to do. 
Though we generally understand our objective
as minimizing edge mistakes, we can also consider maximizing the number of edge satisfactions. These objectives
are equivalent (though they differ in terms of approximations), and just as the LP objective scores
provide a lower bound on edge mistakes, we can derive an upper bound on edge satisfactions
by subtracting the LP objective scores from $|E|$. \looseness=-1 
Figures \ref{fig:logo-alphas}\subref{fig:lo-satisfied-percent} and \ref{fig:logo-alphas}\subref{fig:go-satisfied-percent}
show that both \LLP{} and \GLP{} produce satisfied
edge sets with size at least 99.8\% of optimal. 

We note that these observed approximation factors are especially impressive in the \GLP{} context, where the
theoretical guarantee is not constant but rather a linear function in $b$. Figure \ref{fig:logo-alphas}\subref{fig:go-alphas}
shows that we do not observe
this linear relationship in practice, suggesting that our aforementioned interest in further
study of rounding techniques for the \GOECC{} LP relaxation may bear fruit.

Even more remarkable is the performance of these algorithms on $\beta$. Our empirical
results show that, in nearly every case,
our bicriteria \OECCfull{} algorithms functioned as single-criteria approximations. Indeed, we were unable to
produce a single case in which \GLP{} exceeded its color assignment budget on real-world data. 
\LLP{} performed nearly as
well, producing within-budget approximations on all values of $b$ 
for 5 out of 6 datasets. The \TRIVAGO{} dataset
is the lone exception, though Figure \ref{fig:r-plots}\subref{fig:lo-trivago-betas}
shows that even in this case the algorithm 
outperformed its guarantee for most budgets.

\RLP{} doesn't achieve the near-exactness of the overlapping variants, but it still
significantly outperforms its guarantees. 
Figures \ref{fig:r-plots}\subref{fig:r-alphas} and \ref{fig:r-plots}\subref{fig:r-betas} show that across all datasets and budgets, we never observe
$\alpha > 1.25$ or $\beta > 2$. In fact, in the \COOKING{} and \DAWN{} datasets we see that exceeding the
node deletion budget tends to allow for clusterings with fewer edge mistakes than would otherwise be possible.
We discuss this in more detail in Section \ref{subsec:discussion}.

\begin{table}[b]
    \centering
    \caption{Algorithm runtimes - 
    Reported figures are the maximum across all tested budgets.
    Algorithms and budget values are as described in Table \ref{tab:summary}.
    Experiments
    were run on a machine with an Intel(R) Core(TM) i7-10700 CPU @ 2.90GHz (8 cores)
    and
    64 GB of RAM.}
    \noindent\begin{tabular}{@{}l l l l l l l l}
        \multicolumn{1}{c}{} & \multicolumn{6}{c}{Max Runtime (seconds)}\\
        \cline{2-7}
        \emph{Dataset} & \LLP{} & \LG{} & \GLP{} & \GG{} & \RLP{} & \RG{}\\
        \hline
        \BRAIN{}          &5.68 &0.03 &5.53 &0.03 &5.43 &0.0 \\ 
        \MAG10{}          &9.06 &0.76 &10.39 &0.92 &32.96 &0.81 \\
        \COOKING{}         &60.38 &0.33 &90.09 &0.35 &276.0 &0.08 \\
        \DAWN{}            &5.09 &0.03 &7.91 &0.16 &19.67 &0.04 \\
        \WALMART{}         &307.23 &3.23 &643.39 &3.78 &6968.25 &0.97 \\
        \TRIVAGO{}         &120.95 &12.14 &119.09 &15.35 &206.33 &4.26 \\
        \hline
    \end{tabular}
    \label{tab:runtimes}
    
\end{table}

We conclude this section by highlighting the runtime performance of our algorithms.
Table \ref{tab:runtimes} gives runtimes for each of our greedy and LP-rounding algorithms.
Given their simplicity, it is not surprising that each of
our greedy algorithms is extremely fast. The LP-rounding algorithms have more surprising results.
Amazingly, both \LLP{} and \GLP{}
produce optimal or very nearly optimal solutions very quickly on
a standard workstation with 64 GB of RAM - the sort of machine widely available at nearly
every university or private company. Indeed, even with hundreds of thousands of nodes
and hyperedges, \LLP{} never requires more than  ${\sim}6$ minutes, and \GLP{} never requires more than
${\sim}11$. On the same machine, \RLP{} generally has similar runtime performance. \WALMART{} is
the exception, but even in this case \RLP{} requires only ${\sim}2$ hours. 

\vspace{-.1cm}
\subsection{Discussion}
\label{subsec:discussion}

In this section we present several case studies as a means to answering the following
question: ``what questions should we ask about the structure of our data to understand
how \ECC{} algorithms will perform?'' We analyze performance in terms of both the
problem objectives and a related notion of cluster quality. Lastly, we compare the clusters
generated by our two \OECCfull{} models.

\begin{figure}%
    \begin{subfigure}[t]{.49\linewidth}
        \includegraphics[width=\linewidth]{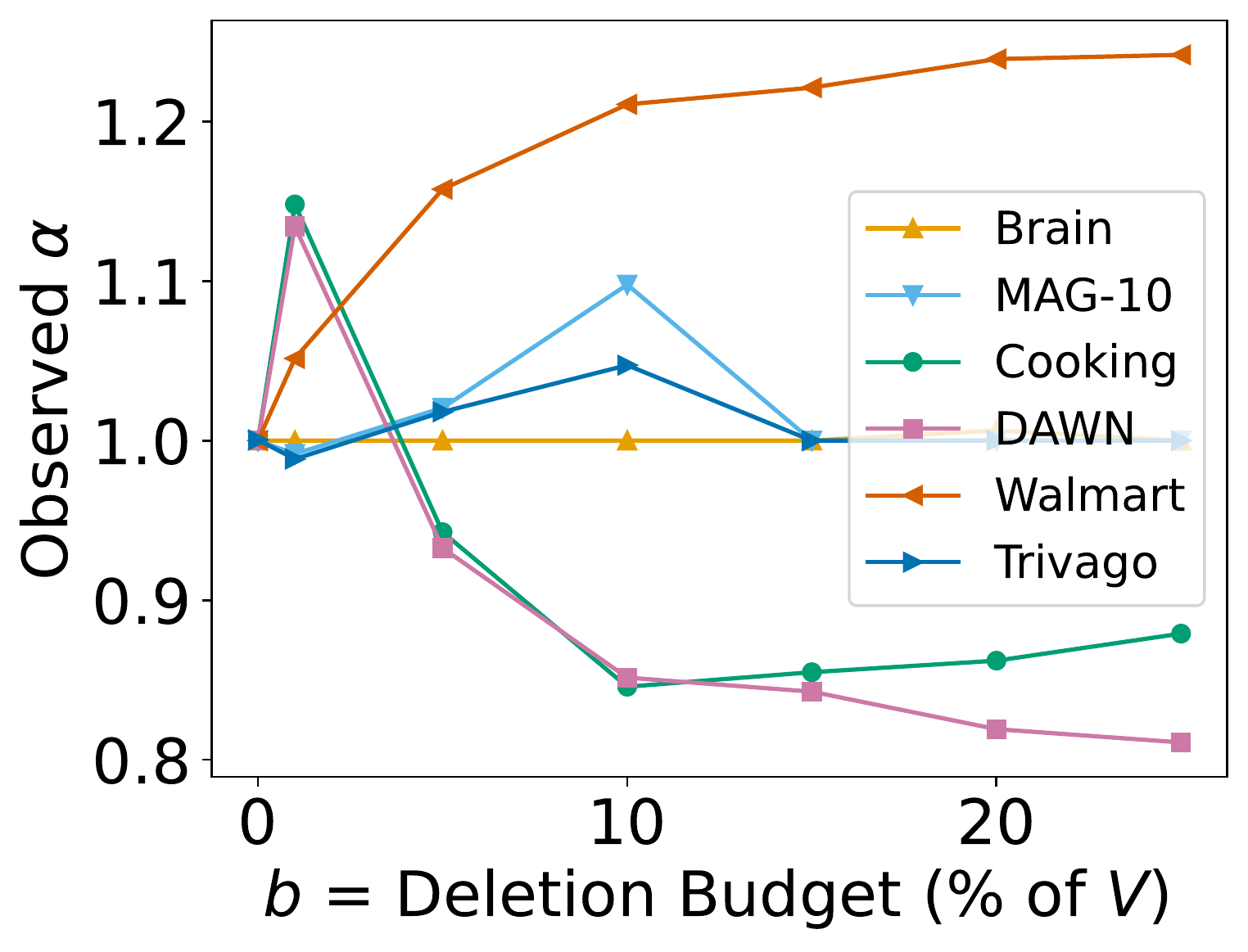}
        \caption{\RLP{} Observed $\alpha$}
        \label{fig:r-alphas}
    \end{subfigure}
    \begin{subfigure}[t]{.49\linewidth}
        \includegraphics[width=\linewidth]{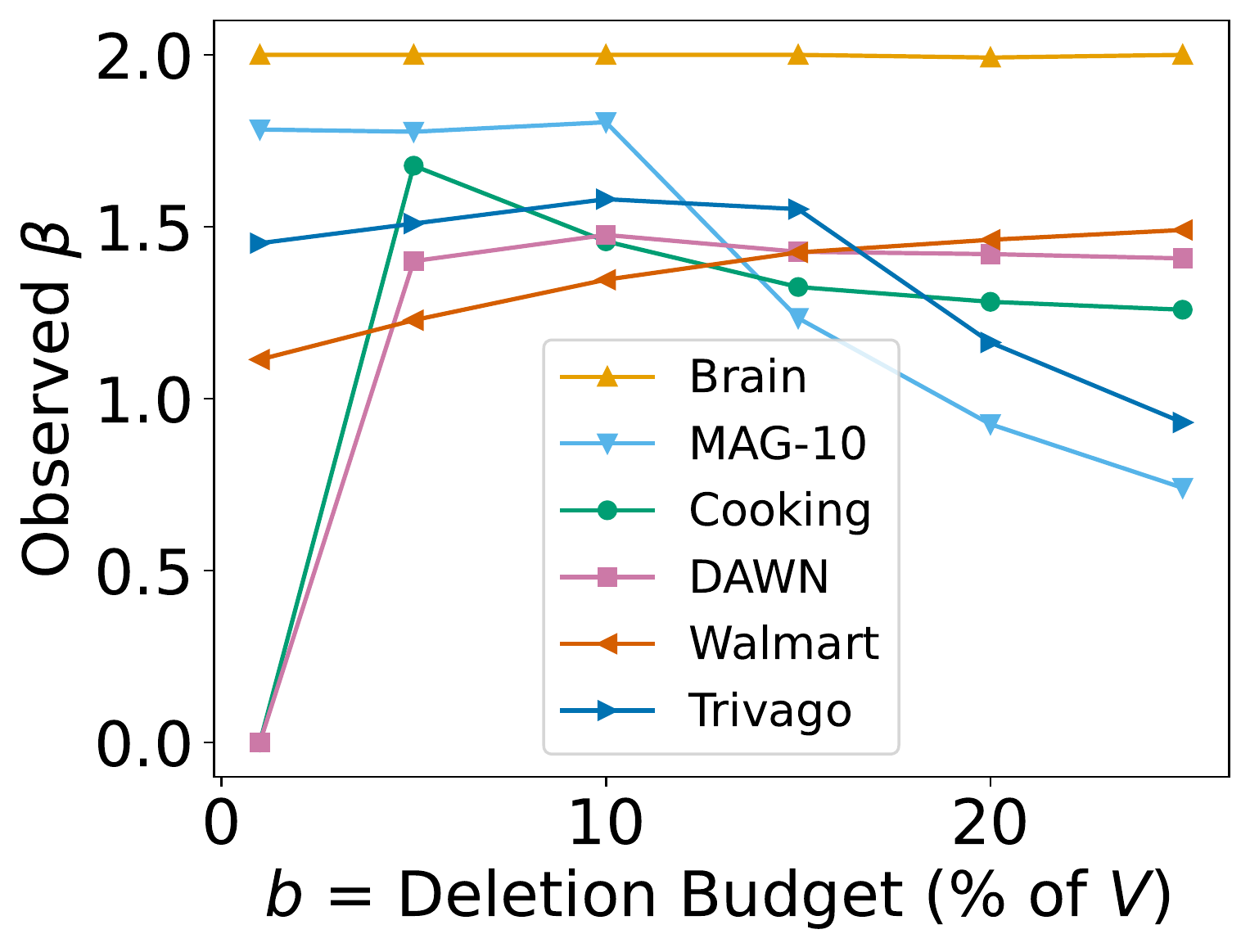}
        \caption{\RLP{} Observed $\beta$}
        \label{fig:r-betas}
    \end{subfigure}\hspace{\fill} 
    \bigskip
    \begin{subfigure}[t]{.49\linewidth}
        \includegraphics[width=\linewidth]{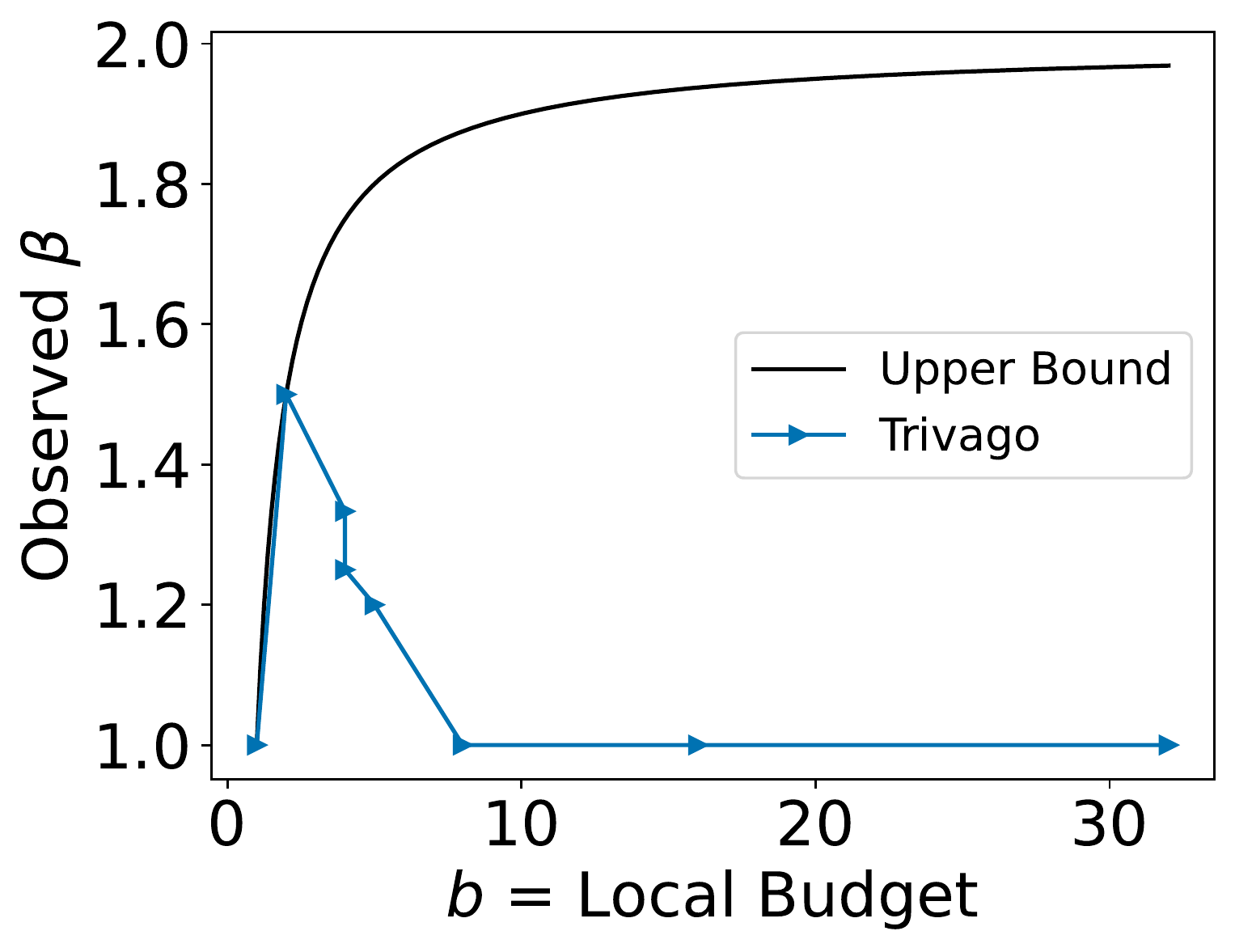}
        \caption{\LLP{} Observed $\beta$}
        \label{fig:lo-trivago-betas}
    \end{subfigure}\hspace{\fill} 
    \begin{subfigure}[t]{.49\linewidth}
        \includegraphics[width=\linewidth]{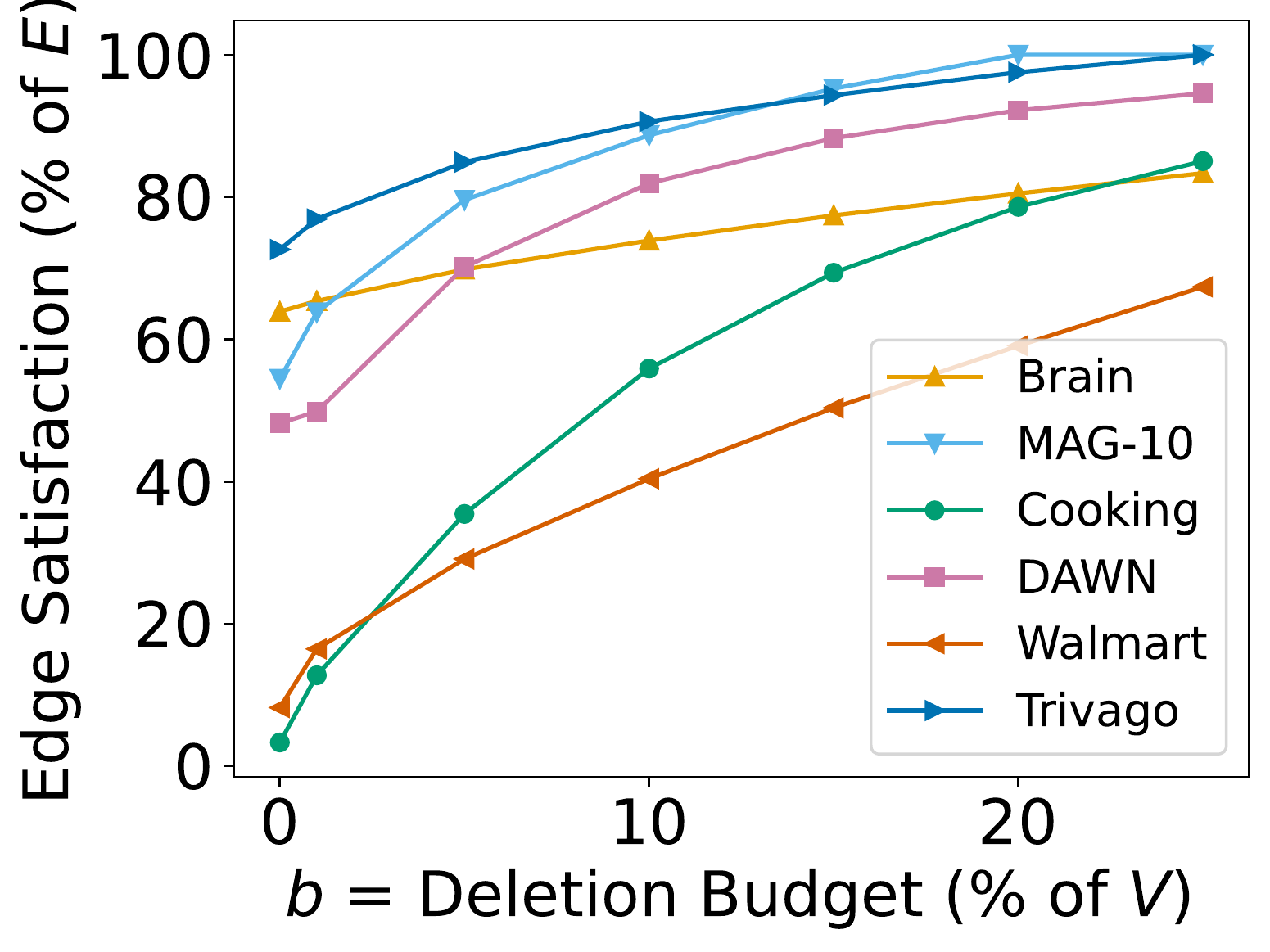}
        \caption{\RG{} Edge Satisfaction}
        \label{fig:r-greedy-satisfaction}
    \end{subfigure}
    \caption{(a)-(b): Observed \RLP{} $\alpha$ and $\beta$ values. $\alpha$ values less than
    1 indicate that the rounded clustering has fewer edge mistakes than is possible
    without violating the node deletion budget, while $\beta$ values less than 1 indicate
    that less than the full budgeted allotment of nodes were deleted. 
    (c): Observed \LLP{} $\beta$ values for the \TRIVAGO{} dataset. The black line is the
    upper \looseness=-1 bound, $2 - 1/b$. (d): Edge satisfaction percentages
    for \RG{}.}%
    \label{fig:r-plots}%
\end{figure}

\textbf{Understanding \RLP{} $\alpha$ Values}.
Figure \ref{fig:r-plots}\subref{fig:r-alphas} raises the following question: why is RLP able to outperform the LP lower bound
for \COOKING{} and \DAWN{}, but not for other datasets? The answer lies in the
structure of the data. A first hint comes from Table \ref{tab:summary}: the mean chromatic degrees in these two
datasets are notably higher than in all others. A more detailed explanation requires new
machinery. We define the \emph{non-dominant degree} of vertex $v$, denoted $d_v^{nd}$, as the number
of (hyper)edges containing $v$ which are \emph{not} $v$'s most frequent edge color. Formally, we
recall the notation of Section \ref{sec:greedy-approximations} and define $d_v^{nd} = d_v - \numvc{v}{\permvi{v}{1}}$.
The \emph{non-dominant degree percentage} of node $v$ is the quotient
$d_v^{nd\%} = d_v^{nd}/d_v$.
\RECC{} can only assign one color to each node, so we should expect that the set of deleted nodes
will tend to have high non-dominant degree and non-dominant degree percentage.
In fact, the
former is precisely the greedy criterion used by \RG{},
and Figure \ref{fig:r-plots}\subref{fig:r-greedy-satisfaction} 
shows that increasing the deletion budget is especially helpful for \RG{} on \COOKING{}.
Table \ref{tabel:degree-stats} gives more detail:
\COOKING{} and \DAWN{} are outliers in that relatively many of their nodes have degree
well-spread across multiple colors. 

Though we have presented this information formally, an intuitive
understanding of the data should be sufficient to predict this structural difference. It makes sense that
many food ingredients are used widely across multiple cuisines, and that many drugs lead commonly to
multiple emergency-room outcomes, based on their common combinations. Conversely, we should expect that
vacation rentals, while perhaps receiving interest from multiple countries, attract the vast
majority of their interest from their own country. We note that \BRAIN{} is an interesting case as well. 
Table \ref{tabel:degree-stats} shows that it has the least heavy-tailed distribution of non-dominant degree. Additionally, it is an
outlier among our datasets in that it is our only (non-hyper) graph, it has only two colors, and it has by far the highest
edge density. Further study is needed to understand how the structure of \BRAIN{} impacts algorithm
performance.

\begin{table}[]
    \centering
    \caption{Network Structure - For each dataset we present maximum, mean, and median
    non-dominant degree, as well as the
    fraction of nodes with chromatic degree $> 1$, and with non-dominant degree percentage $d^{nd\%}$
    at least $5\%$ and $10\%$.}
    \noindent\begin{tabular}{@{}l l l l l l l l}
        \emph{Dataset} & max & mean & med. & $d^{\chi} > 1$ & $d^{nd\%} \geq 5\%$ & $\geq 10\%$\\
        \hline
        \BRAIN{}          &67 &13.74 &12 &91.54 &80.56 &68.50 \\
        \MAG10{}          &96 &0.51 &0 &18.51 &18.50 &18.43 \\
        \COOKING{}         &14595 &38.56 &1 &61.32 &60.56 &58.92 \\
        \DAWN{}            &9944 &84.65 &3 &74.40 &74.02 &72.40 \\
        \WALMART{}         &4379 &2.90 &1 &62.63 &52.62 &52.61 \\
        \TRIVAGO{}         &222 &0.76 &0 &23.27 &22.99 &22.18 \\
        \hline
    \end{tabular}
    \label{tabel:degree-stats}
\end{table}

\textbf{Cluster Quality: \LLP{} vs \LG{}}.
Table \ref{tab:summary} tells us that, in terms of $\alpha$, the worst case performances of our LP-rounding
algorithms are always better than those of their greedy counterparts. In fact we can say
something stronger: across all of our datasets and
both \OECCfull{} models, the LP-rounding
algorithms \emph{always} minimized edge mistakes at least as well as their greedy counterparts,
oftentimes by wide margins.
Counting edge mistakes,
however, is not the only way to measure cluster quality. We now ask whether our LP-rounding
algorithms maintain their superior performance when evaluated by other metrics, using
\LLP{} and \LG{} as an example case. To this end, we call a node $v$ \emph{unused}
by a clustering if it participates in zero satisfied edges. A low number of unused
nodes indicates that a method has
produced dense clusters, i.e., clusters where members combine well. 
Figure \ref{fig:cluster-quality-plots}\subref{fig:llplg-unused} tells us that
in addition to the edge mistake objective, \LLP{} also outperforms LG with respect to cluster
density, as indicated by unused nodes. We ask, however, why this difference in performance is
variable across datasets. Once again, basic intuition about the datasets and algorithms in question
can help us predict the results. Recall that a hyperedge of color $c$ is unsatisfied if \emph{any} of its nodes
is not colored $c$. Consequently, a greedy assignment of color $c$ to node $v$ has more chances
to become irrelevant if the hyperedges of color $c$ that contain $v$ tend to be large. 
Figure \ref{fig:cluster-quality-plots}\subref{fig:llplg-unused}
supports this intuition: \LLP{} outperforms \LG{} with respect to unused nodes especially well
on the \COOKING{} and \WALMART{} datasets, which tend to have large hyperedges, while the performance
gap is lesser on datasets with small hyperedges. 

\textbf{Comparison of \GLP{} and \LLP{}}.
We conclude with a brief comparison of our two models for \OECCfull{}. If the budget $b$ is viewed as
an average per node, then we can compare these models directly, i.e., a local budget $b$
is compared to a global budget $(b - 1)\cdot|V|$. Viewed in this
way, \GOECC{} is a generalization of \LOECC{}, because for such budgets any solution to the
latter is also feasible for the former. Thus, whenever budgets are aligned
\GLP{} will satisfy a larger percentage of edges (Figure \ref{fig:cluster-quality-plots}\subref{fig:golo-satisfaction}).
We note that in real data the margins are quite large: oftentimes \GLP{} satisfies an additional
$10 - 40+$ percent of edges. Figure \ref{fig:cluster-quality-plots}\subref{fig:golo-mistakes} further emphasizes this
point from an alternate perspective: whenever the budget allows for overlap,
using \GLP{} instead of \LLP{} reduces the edge mistake
objective score by (oftentimes much) more than $40\%$.
 We might also ask whether the increased freedom granted
to \GOECC{} results in higher quality clusters by other measures. In particular, 
Figure \ref{fig:cluster-quality-plots}\subref{fig:golo-unused} shows that
\GLP{} generally outperforms \LLP{} with respect to unused nodes. It is encouraging that the increased
freedom granted to \GOECC{} results in higher quality clusters, even measured by a metric that we
do not directly optimize for.

\begin{figure}%
    \begin{subfigure}[t]{.49\linewidth}
        \includegraphics[width=\linewidth]{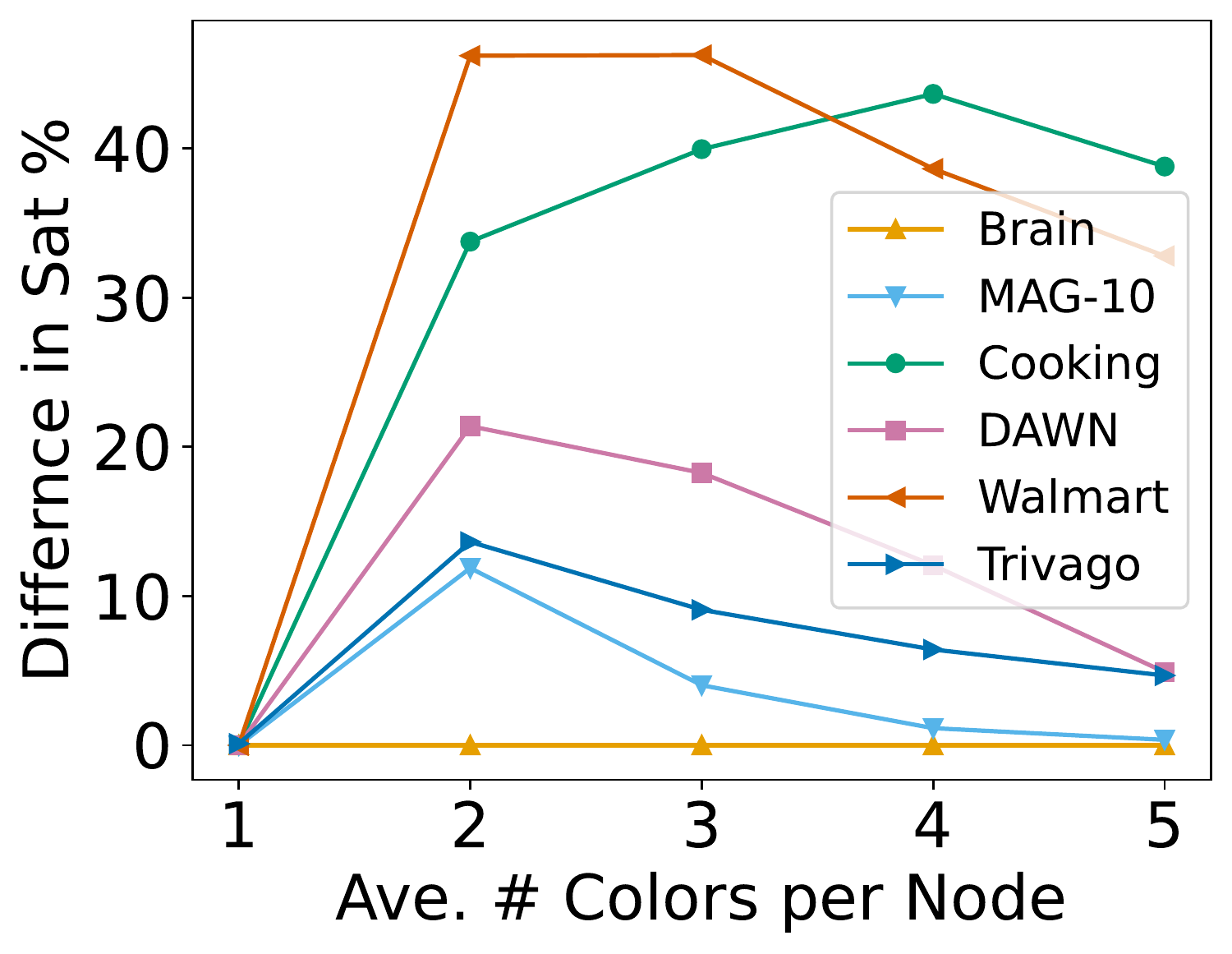}
        \caption{Satisfaction: \GLP{} vs \LLP{}}
        \label{fig:golo-satisfaction}
    \end{subfigure}\hspace{\fill} 
    \begin{subfigure}[t]{.49\linewidth}
        \includegraphics[width=\linewidth]{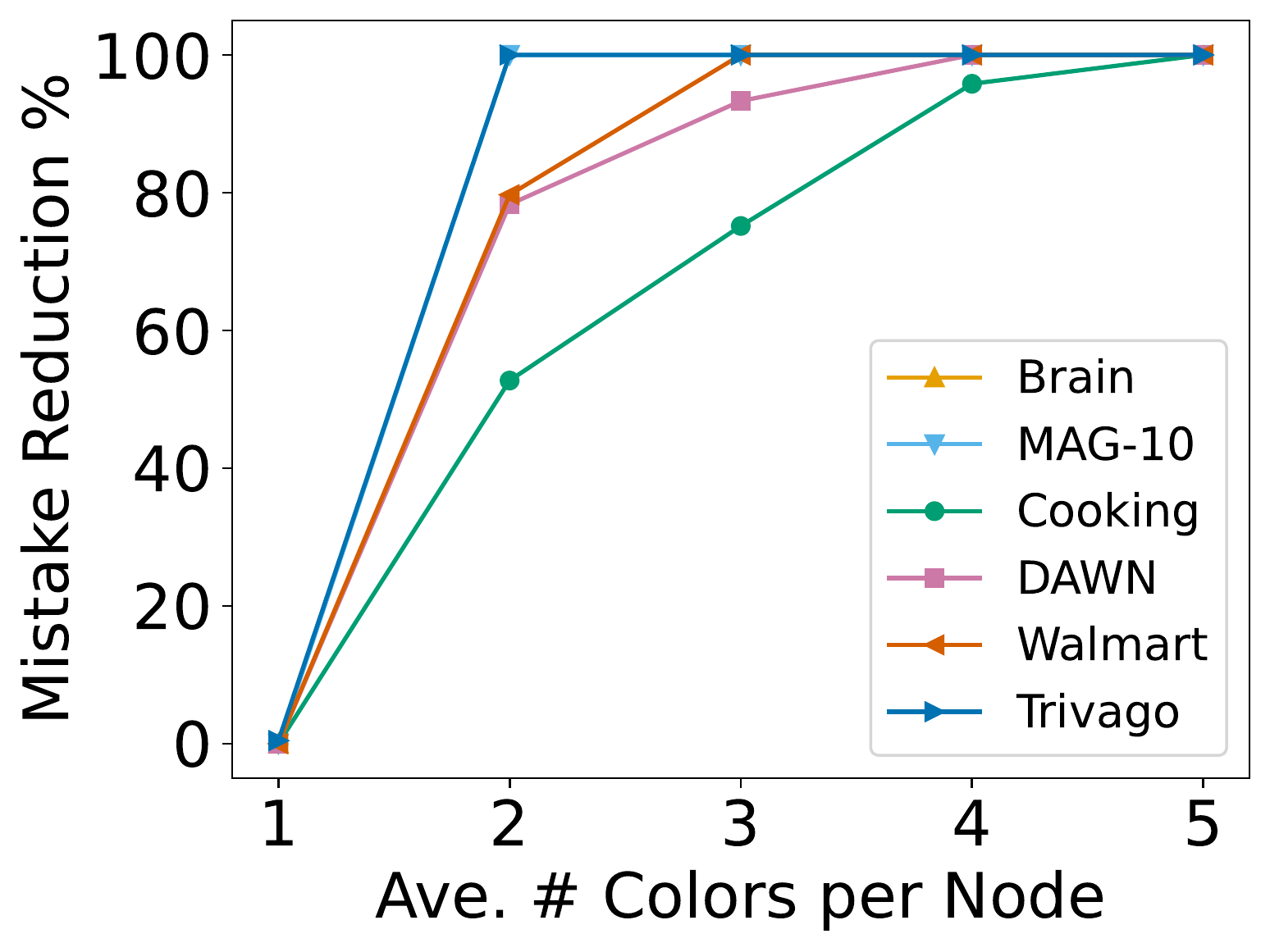}
        \caption{Mistakes: \GLP{} vs \LLP{}}
        \label{fig:golo-mistakes}
    \end{subfigure}
    \bigskip
    \begin{subfigure}[t]{.49\linewidth}
        \includegraphics[width=\linewidth]{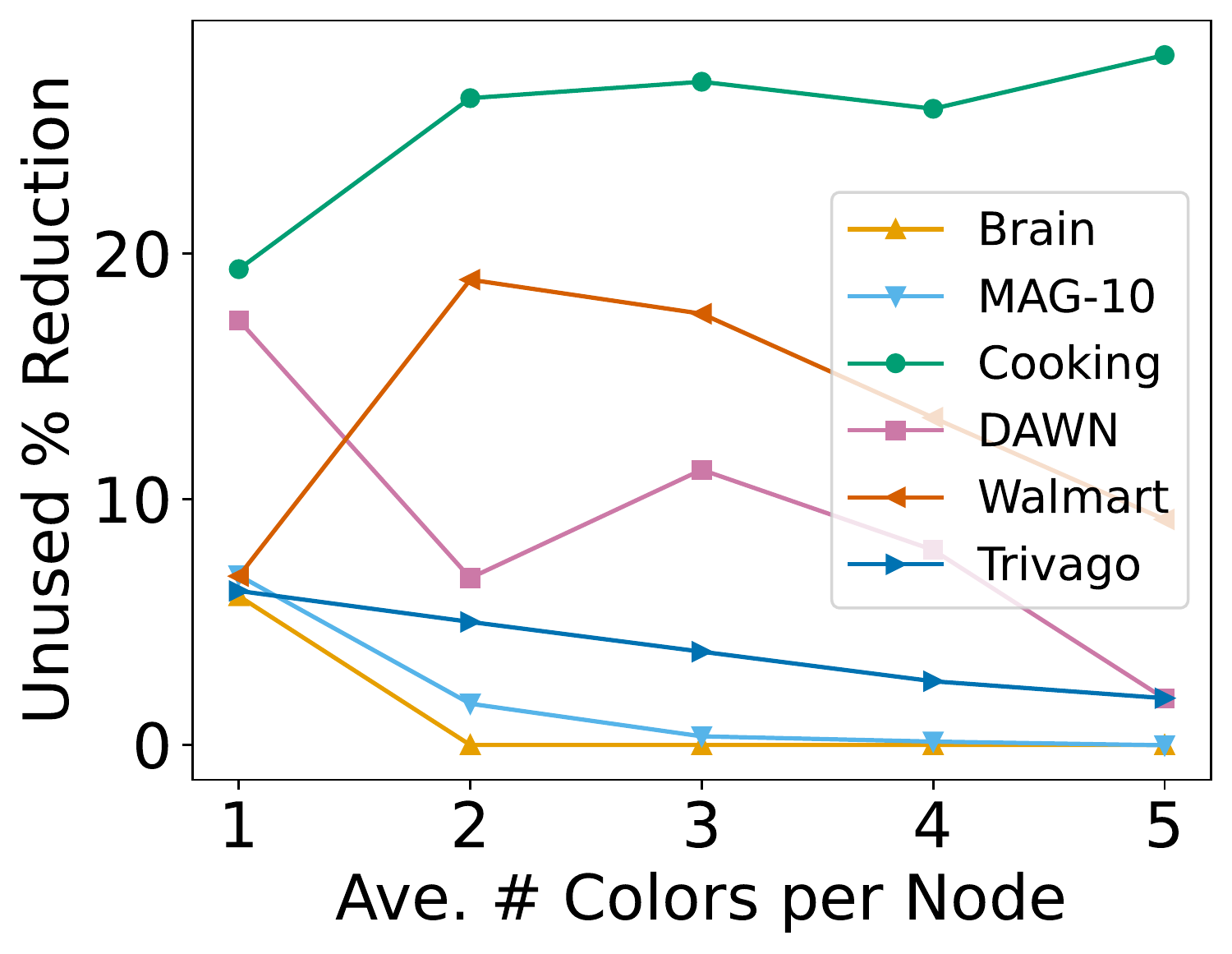}
        \caption{Unused: \LLP{} vs \LG{}}
        \label{fig:llplg-unused}
    \end{subfigure}\hspace{\fill} 
    \begin{subfigure}[t]{.49\linewidth}
        \includegraphics[width=\linewidth]{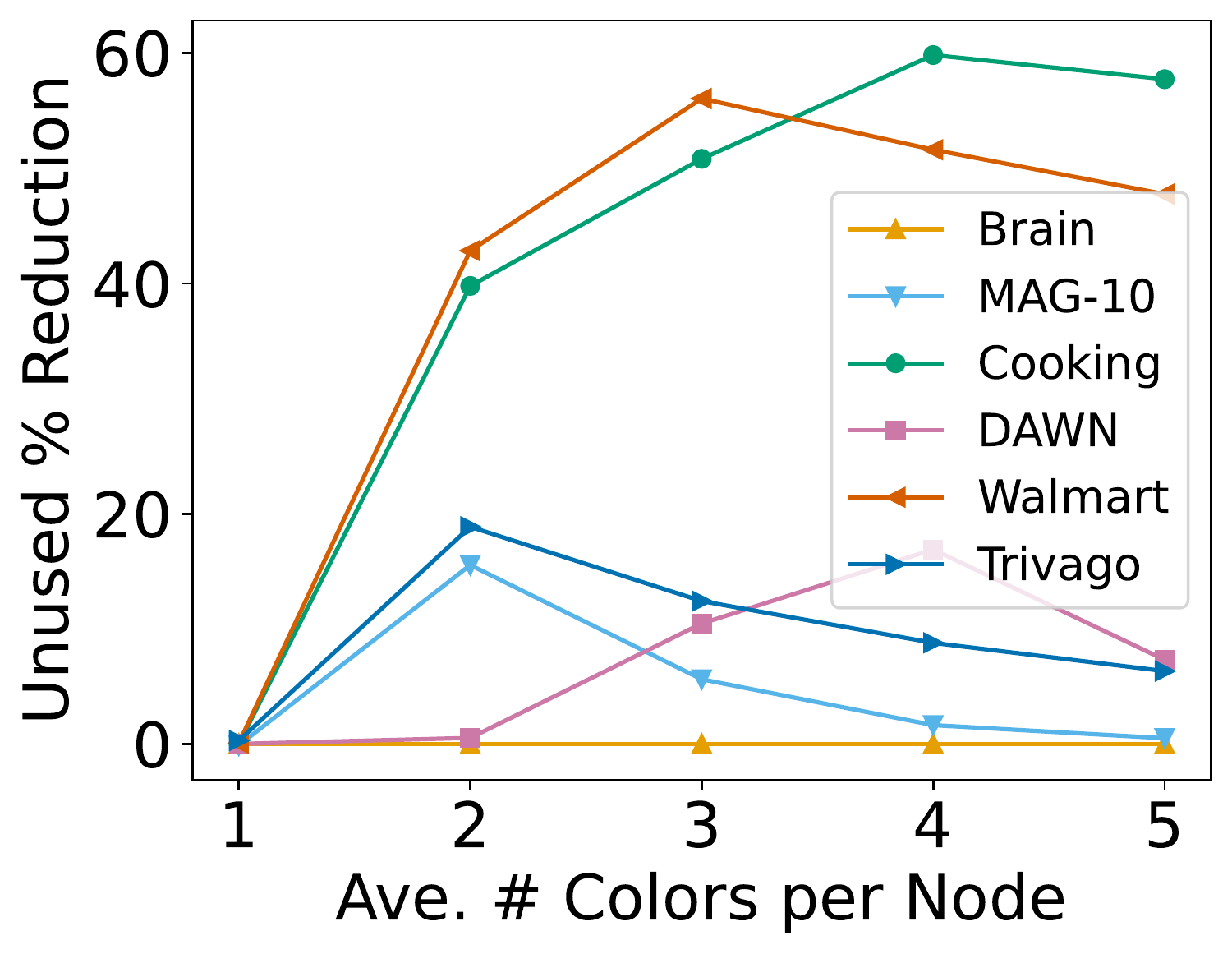}
        \caption{Unused: \GLP{} vs \LLP{}}
        \label{fig:golo-unused}
    \end{subfigure}
    \caption{(a): Absolute difference between percent \GLP{} and \LLP{} edge satisfaction percentages, 
    with same (average) budget per node. (b): Percent reduction (relative difference) in mistakes
    when using \GLP{} instead of \LLP{}. (c): Percent reduction in unused
    nodes when using \LLP{} instead of \LG{}. (d) Percent reduction in
    unused nodes when using \GLP{} instead of \LLP{}. }%
    \label{fig:cluster-quality-plots}%
\end{figure}

%% file: sections/conclusion.tex
\section{Conclusion}

We have addressed two shortcomings of existing categorical (hyper)graph
clustering methods by presenting three generalizations of \ECCfull{}, allowing for either
overlapping clusters or robustness to noisy points. We have addressed the parameterized
complexity of each problem, as well as providing greedy single-criteria and LP-rounding-based
bicriteria approximation algorithms. Concrete questions remain. For example, are there
constant-factor single-criteria approximations for \LOECC{} and \RECC{}, and is there a constant-constant bicriteria
approximation for \GOECC{}? More generally, bicriteria inapproximability is a relatively unexplored topic, and
it would be interesting to develop the theory of hardness in this area.
From the perspective of parameterized algorithms, it is open to determine whether
polynomial kernels exist for our problems in the parameter $t+b$, as well as to identify measures of
edge-colored hypergraph structure which can provably explain the real-world performance of our approximation algorithms.